\RequirePackage{fix-cm}
\documentclass[smallcondensed]{svjour3}     
\smartqed  
\usepackage{graphicx}
\usepackage{mathptmx}      
\usepackage{amsmath}
\usepackage{amssymb}
\usepackage{amsxtra}

\usepackage{amsthm}
\usepackage{tikz,pgfplots,pgfplotstable}
\usetikzlibrary{plotmarks,matrix,chains,scopes,fit,calc,shapes,positioning,decorations,intersections,fit,backgrounds}
\usepackage{xcolor}
\definecolor{light-gray}{gray}{0.75}
\usepackage{wrapfig}
\usepackage{url}
\usepackage{array}
\usepackage{comment}

\usepackage[ruled,linesnumbered,vlined,titlenumbered]{algorithm2e}
\SetKw{KWor}{or}

    \newcommand{\fig}[1]{Fig.~\ref{#1}}           
    \newcommand{\tab}[1]{Table~\ref{#1}}            
    \newcommand{\lem}[1]{Lemma~\ref{#1}}    
    \newcommand{\defref}[1]{Definition~\ref{#1}}
        
    \newcommand{\algo}[1]{Algorithm~\ref{#1}}
    
    \newcommand{\theoref}[1]{Theorem~\ref{#1}}
    \newcommand{\sect}[1]{Section~\ref{#1}}

    \newcommand{\stepref}[1]{Step~\ref{#1}}

    \newtheorem{mydef}{Definition}

    \theoremstyle{remark}
    \newtheorem*{myremark}{Remark}

    \newcommand{\mybinom}[2]{\left(\begin{array}{@{}c@{\,}} #1\\#2 \end{array}\right)}
	\DeclareMathOperator{\wt}{wt_H}

	\newcommand{\floor}[1]{\ensuremath{\left\lfloor #1 \right\rfloor}}
	
	\usepackage{array}
	\newcolumntype{C}{>{$}c<{$}} 
	\newcolumntype{L}{>{$}l<{$}} 
	\newcolumntype{M}[1]{>{$}m{#1} <{$}} 

	\newcommand{\code}[1]{\mathcal{C}_{#1}}
	
	\newcommand{\synd}[1]{\mathbf{S}_{#1}}	
	\newcommand{\rv}[1]{\mathbf{r}_{#1}}	
	\newcommand{\bl}[1]{\mathbf{B}_{#1}}	
	\newcommand{\sv}[1]{\mathbf{b}^{#1}}	
	\newcommand{\mask}[1]{\mathbf{M}_{#1}}	
	
	\newcommand{\field}[1]{\mathbb{F}_{#1}}	
	\newcommand{\matr}[3]{\mathbf{#1}_{#2 \times #3}}
	\newcommand{\berout}{\text{BER}_{out}}
	\newcommand{\berin}{\text{BER}_{in}}

\begin{document}


\title{Improved Decoding and Error Floor Analysis of Staircase Codes\thanks{L. Holzbaur's and A. Wachter-Zeh's work was
    supported by the Technical University of Munich--Institute for Advanced Study, funded by the German Excellence
    Initiative and European Union Seventh Framework Programme under Grant Agreement No. 291763 and the German Research
    Foundation (Deutsche Forschungsgemeinschaft, DFG) unter Grant No. WA3907/1-1.}
}


\author{Lukas Holzbaur \and Hannes Bartz \and Antonia Wachter-Zeh}

\authorrunning{Holzbaur, Bartz, Wachter-Zeh} 

\institute{Lukas Holzbaur \and Antonia Wachter-Zeh\at
  Technical University of Munich\\
  Institute for Communications Engineering\\
  Theresienstr. 90\\
  80333 Munich \\
  Tel.: +49-89-289 \{29052, 23495\}\\
  \email{\{lukas.holzbaur, antonia.wachter-zeh\}@tum.de} 
  \and Hannes Bartz \at
  German Aerospace Center \\
  Institute of Communications and Navigation \\
  Satellite Networks\\
  M\"unchner Stra\ss{}e 20\\
  82234 Oberpfaffenhofen-Wessling\\
  Tel.: +49-8153-282252\\
  Fax:+49-8153-282844\\
  \email{hannes.bartz@dlr.de} }

\date{Received: date / Accepted: date}

\maketitle

\begin{abstract}
  Staircase codes play an important role as error-correcting codes in optical communications. In this paper, a
  low-complexity method for resolving stall patterns when decoding staircase codes is described. Stall patterns are the
  dominating contributor to the error floor in the original decoding method. Our improvement is based on locating stall
  patterns by intersecting non-zero syndromes and flipping the corresponding bits. The approach effectively lowers the
  error floor and allows for a new range of block sizes to be considered for optical communications at a certain code
  rate or, alternatively, a significantly decreased error floor for the same block
  size. 
  Further, an improved error floor analysis is introduced which provides a more accurate estimation of the contributions
  to the error floor.  
\end{abstract}

\newpage

\section{Introduction}

Staircase codes were introduced by Smith \emph{et al.} in~\cite{stairfec} and are a powerful code construction based on
a binary Bose-Ray-Chaudhuri-Hocquenghem (BCH) component code, designed for error-correction in high-speed optical
communication systems. With performance close to the capacity of the binary symmetric channel (BSC) for high rates and
decoder complexity lower than a comparable low-density parity-check (LDPC) code, staircase codes provide a
cost-efficient alternative to soft decision decoding of LDPC codes. As shown in~\cite{staircase633}, staircase codes
perform well for a multitude of different parameters. However, the usability of staircase codes is limited by the
requirement of optical communication systems to guarantee an error floor below $10^{-15}$, which allows small block
sizes of the staircase code only at relatively low code rates.

Similar to trapping sets in decoding LDPC codes, certain constellations of errors, called \textit{stall patterns},
cannot be resolved by the component codes of the staircase code. A strategy that enables the decoder to resolve stall
patterns will improve the performance in the error floor region and potentially allow for more efficient decoding, as
smaller block sizes can be used. For the structurally closely related product, half-product, and braided codes, several
approaches for resolving stall patterns have been proposed. In~\cite{HPFM,jian} the resolving of stall patterns by
erasure decoding is considered. Compared to the approaches based on bit-flipping~\cite{patent,stapaproduct,IHPC} erasure
decoding has the advantage of only requiring one iteration. However, especially when considering large stall patterns,
it is shown to be outperformed by bit-flipping. To evaluate the performance of a code with given parameters, an analysis
of the error floor is required. While~\cite{stapaproduct,IHPC} offer an analysis based on exhaustive search, this work
extends the analytical approach of \cite{stairfec} by relating the problem of counting stall patterns to the numerical
problem of finding the number of binary matrices with certain row and column weight~\cite{mat01} to obtain a
significantly more accurate estimation.

In this paper, we present an improved decoder for staircase codes, which is able to locate stall patterns and resolves
many of them by adapting and extending the concept of bit-flipping~\cite{patent,stapaproduct,IHPC}. We show that
bit-flipping can guarantee to correct \emph{all} stall patterns when the number of involved columns and rows is each
smaller than the minimum distance of the component BCH code and no undetected error events occur. Another contribution
of this work is a new estimation of the error floor that is significantly more accurate than the one
from~\cite{stairfec}. Finally, we present conjectures on the performance obtained by combining estimation and simulation
for a staircase code with a quarter of the block size compared to the scheme of~\cite{stairfec}. These show that an
output bit error rate of~${\berout = 10^{-15}}$ is reached at~$\sim 1$~dB from BSC capacity at the cost of a small rate
loss ($\frac{236}{255}$ compared to~$\frac{239}{255}$).
This scheme is estimated to achieve a net coding gain (NCG) of $9.16$dB at a $\berout$ of $10^{-15}$.

Parts of this work have recently been presented at WCC 2017 \cite{Holzbaur2017}. In this paper, we give more details
on the process of resolving stall patterns and on the error floor analysis. Further we introduce a
computationally less complex method for estimating the exact number of stall patterns, based on a combination of
analysis and simulation.

\section{Staircase Codes} \label{sec:staircasecodes}

\subsection{Encoding} \label{subsec:stcencoding}

\begin{figure}
  \begin{center}

\def\x{1.5} 

\begin{tikzpicture}

	\fill[blue!40!white] (\x*2.75,\x*4) rectangle (\x*3,\x*3);
	\fill[blue!40!white] (\x*2,\x*2.25) rectangle (\x*3,\x*2);
	\fill[blue!40!white] (\x*3.75,\x*3) rectangle (\x*4,\x*2);
	\fill[blue!40!white] (\x*3,\x*1.25) rectangle (\x*4,\x*1);
	\fill[blue!40!white] (\x*4.75,\x*2) rectangle (\x*5,\x*1);

	\draw (\x*1,\x*4) rectangle (\x*2,\x*3);
	\draw (\x*2,\x*4) rectangle (\x*3,\x*3);
	\draw (\x*2,\x*3) rectangle (\x*3,\x*2);
	\draw (\x*3,\x*3) rectangle (\x*4,\x*2);
	\draw (\x*3,\x*2) rectangle (\x*4,\x*1);
	\draw (\x*4,\x*2) rectangle (\x*5,\x*1);

	\node[draw=none,font=\LARGE] at (\x*1.5,\x*3.5)  {$B_0$};
	\node[draw=none,font=\LARGE] at (\x*2.5,\x*3.5)  {$B_1$};
	\node[draw=none,font=\LARGE] at (\x*2.5,\x*2.5)  {$B_2^T$};
	\node[draw=none,font=\LARGE] at (\x*3.5,\x*2.5)  {$B_3$};
	\node[draw=none,font=\LARGE] at (\x*3.5,\x*1.5)  {$B_4^T$};
	\node[draw=none,font=\LARGE] at (\x*4.5,\x*1.5)  {$B_5$};
	
	\draw (\x*4,\x*1) -- (\x*4,\x*0.7);
	\draw (\x*5,\x*1) -- (\x*5,\x*0.7);	
	\draw (\x*4.5,\x*0.9) node[circle,fill,inner sep=\x*0.2pt]{};
	\draw (\x*4.5,\x*0.835) node[circle,fill,inner sep=\x*0.2pt]{};
	\draw (\x*4.5,\x*0.77) node[circle,fill,inner sep=\x*0.2pt]{};

	\fill[blue!40!white,anchor=west] (\x*3.75,\x*3.25) rectangle (\x*3.9,\x*3.4);
	\node[draw=none,anchor=west] at (\x*3.9,\x*3.31)  {Parity bits};
	
	\draw[thick,<-] (\x*1,\x*4.25) -- (\x*1.85,\x*4.25);
	\draw[thick,<-] (\x*3,\x*4.25) -- (\x*2.15,\x*4.25);
	\node[draw=none] at (\x*2,\x*4.25) {$n$};
	
	\draw[thick,<-] (\x*1,\x*4.1) -- (\x*1.725,\x*4.1);
	\draw[thick,<-] (\x*2.75,\x*4.1) -- (\x*2.025,\x*4.1);
	\node[draw=none] at (\x*1.8725,\x*4.1) {$k$};
	
	\draw[thick,<-] (\x*0.8,\x*4) -- (\x*0.8,\x*3.65);
	\draw[thick,<-] (\x*0.8,\x*3) -- (\x*0.8,\x*3.35);
	\node[draw=none] at (\x*0.8,\x*3.5) {$m$};
	
	\draw[thick,<-] (\x*1,\x*2.8) -- (\x*1.35,\x*2.8);
	\draw[thick,<-] (\x*2,\x*2.8) -- (\x*1.65,\x*2.8);
	\node[draw=none] at (\x*1.5,\x*2.8) {$m$};

	\draw[thick,dashed] (\x*1,\x*3.8) -- (\x*3,\x*3.8);
	\node[draw=none,anchor=west] at (\x*3.15,\x*3.8) {$\in$ BCH$(n,k)$};

	\draw[thick,dashed] (\x*2.2,\x*4) -- (\x*2.2,\x*2);
	\node[draw=none,anchor=west,rotate=270] at (\x*2.2,\x*1.85) {$\in$ BCH$(n,k)$};
	
	\node[draw=none] at (\x*2,\x*0.6) {};
	\node[draw=none] at (\x*0.5,\x*4.5) {};

\end{tikzpicture}
  \end{center}
  \caption{Illustration of the structure of a staircase code.}
  \label{fig:stcstructure}
\end{figure}
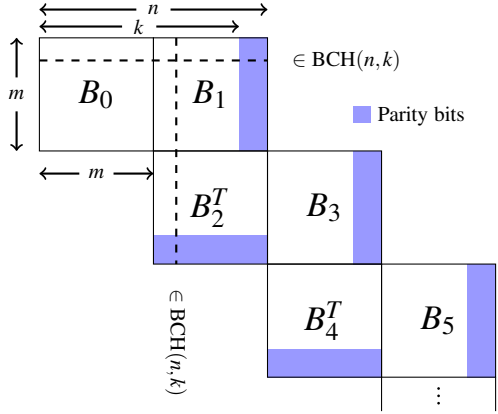

Staircase codes are encoded block-wise, where each block~$\bl{i}$ is a binary $m \times m$ matrix. The encoding
procedure is based on a component code of length $n=2m$, dimension $k>m$ and error-correcting capability $t$. In
\cite{stairfec,Hager20172} and in this work, extended BCH codes are used as component codes. A BCH code is a cyclic code
and is therefore given by all codewords $c(x)$ for which $g(x) | c(x)$, where $c(x)$ is the polynomial representation of
a codeword and $g(x)$ is the generating polynomial. We refer to the code as an extended code if it has even minimum
distance $d\geq 2t+2$, i.e., the code contains only codewords of even weight, which is equivalent to requiring
$(x+1) | g(x)$. Then every codeword can be written as
\begin{align*}
  c(x) = u(x)(x+1)g'(x) = x\cdot u(x) g'(x) + u(x) g'(x),
\end{align*}
which is a sum of two polynomials of the same weight and therefore of even weight. This code is used to encode every row
and column in systematic form.  It follows that, except for the first block which is initialized to all-zeros, every
block consists of $m(k-m) = m\cdot k'$ information bits and $m(n-k) = m(m-k')$ redundancy bits, giving the code rate
\begin{equation*}
  R = \frac{m(k-m)}{m(k-m)+m(n-k)} = \frac{k-m}{n-m} = \frac{k'}{m} .
\end{equation*}
The blocks of a staircase code are defined such that each row of $\left[\bl{i-1}^T \;\; \bl{i}\right]$ is a codeword of
the BCH component code, for all $i \geq 1$. Encoding of block $\bl{i}$ is done by taking the transpose of $\bl{i-1}$ and
appending the $m(k-m)$ information bits of block $\bl{i}$ to obtain an $m \times k$ matrix. Then each row is
encoded with the systematic BCH code to obtain the encoded block $\bl{i}$.  An illustration of the staircase code
structure is given in \fig{fig:stcstructure}.

\subsection{Sliding-Window Decoding} \label{subsec:stcdecoding}

The decoding algorithm of \cite{stairfec} is based on multiple iterations of hard-decision BCH decoders operating on a
sliding window of~$W$ blocks. A window comprised of $\bl{i}$ to $\bl{i+W-1}$ is decoded by first decoding the received
words spanning the rows of $\left[\bl{i}^T \;\; \bl{i+1}\right]$, followed by the codewords spanning the rows of
$\left[\bl{i+1}^T \;\; \bl{i+2}\right]$, until the last block $\bl{i+W-1}$ of the window is reached. Then, the decoder
returns to $\left[\bl{i}^T \;\; \bl{i+1}\right]$ and the process is repeated. If no more errors are detected in the last
block of the window or a fixed maximum number of iterations $v_{\max}$ is reached, the decoder declares $\bl{i}$ as
decoded, slides the window by one block and repeats the process for the new window comprised of $\bl{i+1}$ to
$\bl{i+W}$.

Shortly after our work a new decoding algorithm has been proposed \cite{Hager2017,Hager20172} that uses so-called anchor
codewords, which have likely been decoded correctly, in order to avoid the undesirable undetected error events, also
called miscorrections, in the decoding process. This significantly improves performance in the waterfall region compared
to the decoder proposed in \cite{stairfec}.

\subsection{Stall Patterns and Known Error Floor Analysis} \label{subsec:analysisold}
An $[n,k]$ BCH component code with minimum distance $d_{\min}$ can detect any $d_{\min} -1$ errors and has (unique) error-correcting capability~$t = \floor{\frac{d_{\min} -1}{2}}$, \emph{i.e.}, if $t+1$ errors affect the codeword, the decoder is not able to resolve them. Similar to product codes, the minimum distance between two valid semi-infinite staircase code codewords (comparable to the free distance of convolutional codes) is given by $d_{\code{}} = d_{\min}^2$, following from linearity and the minimal weight of a valid staircase code codeword. As for product and convolutional codes, the minimum/free distance is not a good measure for the error correction capability of the code. However, when considering the error floor, the minimum distance gives a lower bound on the weight of a theoretically undetectable error pattern, which will be important in the analysis presented in this work.

A stall pattern of a staircase code is a set of erroneous bit positions such that each erroneous row and column contains
at least $t+1$ erroneous bits. It follows that the component codes are unable to resolve these errors and the pattern
cannot be resolved, despite possibly having less than $\floor{\frac{d_{\code{}}-1}{2}}$ errors. The minimal number of
rows $K$ and columns $L$ involved in such a pattern is $K = L = t+1$ and the minimal number of errors is
$\epsilon = (t+1)^2$. For such \emph{minimal stall patterns} (compare \fig{fig:minstpa}), every intersection of an
involved row and an involved column is an erroneous bit.

If more than $t+1$ rows or columns are part of the stall pattern, it is possible that not every bit in the intersection
of involved rows and columns is in error.  The weight of the error vectors of each involved row or column has to be at
least $t+1$ and therefore, the number of errors $\epsilon$ in a $(K,L)$ stall pattern is bounded by
\begin{equation}\label{eq:epsbound}
  \epsilon_{\min} \stackrel{\triangle}{=}  \max\left\{K,L\right\} \cdot (t+1) \leq \epsilon \leq K \cdot L .
\end{equation}
\fig{fig:nonminstpa} shows a non-minimal $(4,4)$ stall pattern with $t=2$ and $\epsilon = \epsilon_{\min} = 12$.

\begin{figure}
  \centering
  \begin{minipage}{0.45\textwidth}
    \centering \def\x{1.5} 

\tikzset{cross/.style={cross out, draw=black, fill=none, minimum size=2*(#1-\pgflinewidth), inner sep=0pt, outer sep=0pt}, cross/.default={\x*2pt}}

\begin{tikzpicture}

	\node[draw=none] at (\x*1.5,\x*2.5)  {\color{light-gray}$\bl{i-1}^T$};
	\node[draw=none] at (\x*2.5,\x*2.5)  {\color{light-gray}$\bl{i}$};
	\node[draw=none] at (\x*2.5,\x*1.5)  {\color{light-gray}$\bl{i+1}^T$};
	\node[draw=none] at (\x*3.5,\x*1.5)  {\color{light-gray}$\bl{i+2}$};

	\draw (\x*1,\x*3) rectangle (\x*2,\x*2);
	\draw (\x*2,\x*3) rectangle (\x*3,\x*2);
	\draw (\x*2,\x*2) rectangle (\x*3,\x*1);
	\draw (\x*3,\x*2) rectangle (\x*4,\x*1);
	
	\draw (\x*1,\x*3) -- (\x*1,\x*3.3);
	\draw (\x*2,\x*3) -- (\x*2,\x*3.3);	
	\draw (\x*1.5,\x*3.1) node[circle,fill,inner sep=\x*0.2pt]{};
	\draw (\x*1.5,\x*3.165) node[circle,fill,inner sep=\x*0.2pt]{};
	\draw (\x*1.5,\x*3.23) node[circle,fill,inner sep=\x*0.2pt]{};

	\draw (\x*3,\x*1) -- (\x*3,\x*0.7);
	\draw (\x*4,\x*1) -- (\x*4,\x*0.7);	
	\draw (\x*3.5,\x*0.9) node[circle,fill,inner sep=\x*0.2pt]{};
	\draw (\x*3.5,\x*0.835) node[circle,fill,inner sep=\x*0.2pt]{};
	\draw (\x*3.5,\x*0.77) node[circle,fill,inner sep=\x*0.2pt]{};

	\draw[dashed] (\x*1,\x*2.75) -- (\x*3,\x*2.75);
	\draw[dashed] (\x*2,\x*1.2) -- (\x*4,\x*1.2);
	\draw[dashed] (\x*2,\x*1.8) -- (\x*4,\x*1.8);
	
	\draw[dashed] (\x*2.4,\x*3) -- (\x*2.4,\x*1);
	\draw[dashed] (\x*2.6,\x*3) -- (\x*2.6,\x*1);
	\draw[dashed] (\x*2.8,\x*3) -- (\x*2.8,\x*1);
	
	\draw (\x*2.4,\x*1.2) node[cross]{};
	\draw (\x*2.4,\x*1.8) node[cross]{};
	\draw (\x*2.4,\x*2.75) node[cross]{};
	\draw (\x*2.6,\x*1.2) node[cross]{};
	\draw (\x*2.6,\x*1.8) node[cross]{};
	\draw (\x*2.6,\x*2.75) node[cross]{};
	\draw (\x*2.8,\x*1.2) node[cross]{};
	\draw (\x*2.8,\x*1.8) node[cross]{};

	\draw (\x*2.8,\x*2.75) node[cross]{};

\end{tikzpicture}
    \caption{Minimal stall pattern of size ${(K=3,L=3)}$ for a $t=2$ error-correcting code.}
    \label{fig:minstpa}
  \end{minipage}
  \hspace{0.05\textwidth}
  \begin{minipage}{0.45\textwidth}
    \centering \def\x{1.5} 

\tikzset{cross/.style={cross out, draw=black, fill=none, minimum size=2*(#1-\pgflinewidth), inner sep=0pt, outer sep=0pt}, cross/.default={\x*2pt}}

\begin{tikzpicture}

	\node[draw=none] at (\x*1.5,\x*2.5)  {\color{light-gray}$\bl{i-1}^T$};
	\node[draw=none] at (\x*2.5,\x*2.5)  {\color{light-gray}$\bl{i}$};
	\node[draw=none] at (\x*2.5,\x*1.5)  {\color{light-gray}$\bl{i+1}^T$};
	\node[draw=none] at (\x*3.5,\x*1.5)  {\color{light-gray}$\bl{i+2}$};

	\draw (\x*1,\x*3) rectangle (\x*2,\x*2);
	\draw (\x*2,\x*3) rectangle (\x*3,\x*2);
	\draw (\x*2,\x*2) rectangle (\x*3,\x*1);
	\draw (\x*3,\x*2) rectangle (\x*4,\x*1);
	
	\draw (\x*1,\x*3) -- (\x*1,\x*3.3);
	\draw (\x*2,\x*3) -- (\x*2,\x*3.3);	
	\draw (\x*1.5,\x*3.1) node[circle,fill,inner sep=\x*0.2pt]{};
	\draw (\x*1.5,\x*3.165) node[circle,fill,inner sep=\x*0.2pt]{};
	\draw (\x*1.5,\x*3.23) node[circle,fill,inner sep=\x*0.2pt]{};
	
	\draw (\x*3,\x*1) -- (\x*3,\x*0.7);
	\draw (\x*4,\x*1) -- (\x*4,\x*0.7);	
	\draw (\x*3.5,\x*0.9) node[circle,fill,inner sep=\x*0.2pt]{};
	\draw (\x*3.5,\x*0.835) node[circle,fill,inner sep=\x*0.2pt]{};
	\draw (\x*3.5,\x*0.77) node[circle,fill,inner sep=\x*0.2pt]{};

	\draw[dashed] (\x*1,\x*2.3) -- (\x*3,\x*2.3);
	\draw[dashed] (\x*1,\x*2.75) -- (\x*3,\x*2.75);
	\draw[dashed] (\x*2,\x*1.2) -- (\x*4,\x*1.2);
	\draw[dashed] (\x*2,\x*1.8) -- (\x*4,\x*1.8);
	
	\draw[dashed] (\x*2.1,\x*3) -- (\x*2.1,\x*1);
	\draw[dashed] (\x*2.3,\x*3) -- (\x*2.3,\x*1);
	\draw[dashed] (\x*2.65,\x*3) -- (\x*2.65,\x*1);
	\draw[dashed] (\x*2.9,\x*3) -- (\x*2.9,\x*1);
	
	\draw (\x*2.1,\x*1.8) node[cross]{};
	\draw (\x*2.1,\x*2.3) node[cross]{};
	\draw (\x*2.1,\x*2.75) node[cross]{};
	\draw (\x*2.3,\x*1.2) node[cross]{};
	\draw (\x*2.3,\x*1.8) node[cross]{};
	\draw (\x*2.3,\x*2.75) node[cross]{};
	\draw (\x*2.65,\x*1.2) node[cross]{};
	\draw (\x*2.65,\x*1.8) node[cross]{};
	\draw (\x*2.65,\x*2.3) node[cross]{};
	\draw (\x*2.9,\x*1.2) node[cross]{};
	\draw (\x*2.9,\x*2.3) node[cross]{};
	\draw (\x*2.9,\x*2.75) node[cross]{};

\end{tikzpicture}
    \caption{Non-minimal stall pattern of size ${(K=4,L=4)}$ for a $t\leq2$ error-correcting code.}
    \label{fig:nonminstpa}
  \end{minipage}

\end{figure}

The error floor estimation given in \cite{stairfec} is based on the assumption that the dominating contributors to the
error floor are stall patterns. It is obtained by enumerating the number of possible stall patterns and weighting each
pattern with the probability that the corresponding positions are in error. A stall pattern is associated with~$\bl{i}$
which has lowest index that contains at least one of its errors. The number of combinations of $K$ rows and $L$ columns
such that the stall pattern belongs to a certain block is
\begin{equation} \label{eq:choicesroco} A_{K,L} = \mybinom{m}{L}\cdot \sum_{a=1}^{K} \mybinom{m}{a} \cdot
  \mybinom{m}{K-a} .
\end{equation}
Given the rows and columns, the number of different ways to distribute errors within their intersections is denoted by
$N_{K,L}^\epsilon$ and bounded from above by (see \cite{stairfec})

\begin{equation} \label{eq:amerrdist} N_{K,L}^\epsilon \leq \hat{N}_{K,L}^{\epsilon} =
  \mybinom{\min\left\{K,L\right\}}{t+1}^{\max\left\{K,L\right\}} \cdot \mybinom{K \cdot L - \epsilon_{min}}{\epsilon -
    \epsilon_{\min}} .
\end{equation}

With \eqref{eq:choicesroco} and \eqref{eq:amerrdist}, the contribution of $(K,L)$-stall patterns to the $\berout$ in the
error floor region can be overbounded by weighing each pattern with the probability that the corresponding positions are
in error, to obtain
\begin{equation} \label{eq:bounderrflo} \sum_{\epsilon=(t+1)\cdot \max \left\{K,L\right\}}^{K \cdot L}
  \frac{\epsilon}{m^2} \cdot A_{K,L} \cdot \hat{N}_{K,L}^{\epsilon} \cdot (p + \xi)^\epsilon ,
\end{equation}
where $p$ is the crossover probability of the BSC and $\xi$ is an additional correction factor adjusting for the
occurrence of undetected error events during the iterations of the decoding process. Unfortunately, it is difficult to
give an analytic bound on $\xi$ and current approaches rely on determining the appropriate value via estimation
and/or simulation~\cite{stairfec,Hager20172}.

\section{Resolving Stall Patterns} \label{sec:improdec}

\subsection{The Bit-Flip Operation} \label{subsec:loerrflo}

Assume that all errors that are not part of a stall pattern are resolved by the regular sliding-window decoding
procedure (see \sect{subsec:stcdecoding}) and hence only stall patterns remain.

In a minimal stall pattern, each involved row and column contains exactly $t+1$ erroneous bits which results in
a non-zero syndrome for the component code with distance $d \geq 2t+2$. Thus, a minimal stall pattern can be resolved by
flipping each bit at the intersection of the words with non-zero syndromes.

\begin{mydef}[Bit-Flip]\label{def:syndvec}
  Consider a staircase code with a $t$ error-correcting component code and let $\rv{0}^i,...,\rv{m-1}^i$ be
  the received words corresponding to component codewords with redundancy bits in $\bl{i+1}$. Let
  $\synd{\rv{j}^i}$ be the syndrome of $\rv{j}^i$. Define the elements of the vector
  $\sv{i} \in \mathbb{F}_2^{m\times1}$ for $0 \leq j \leq m-1$ by:
  \begin{equation} \label{eq:syndvec} \sv{i} (j) = \left\{
      \begin{array}{ll}
	1, & \text{if} \;\;\; \synd{\rv{j}^i} \neq \mathbf{0}\\
	0, & \text{else}.
      \end{array} \right.
  \end{equation}
  Let $\mask{i} = \sv{i-1} \cdot (\sv{i})^T \in \mathbb{F}_2^{m \times m}$ be the \textbf{masking matrix} and let
  $\bl{i}^{(z)}$ be block $\bl{i}$ after $z$ decoding iterations. Define the operation \textbf{bit-flip} as
  \begin{equation}
    \bl{i}^{(z+1)} = \bl{i}^{(z)} + \mask{i}. \label{eq:mask}
  \end{equation}
\end{mydef}
The matrix $\mask{i}$ is a binary $m\times m$ block which is non-zero only in the positions involved in a stall pattern
of given size and maximum weight. Assuming no miscorrections, \emph{i.e.}, the syndrome is non-zero for every involved row and
column, it covers all of their intersections.

\subsection{Analysis of Bit-Flip without Undetected Error Events} \label{subsec:resnoundet}

In this section, we analyze the performance of the bit-flip operation under the assumption that no \emph{undetected
  error events} occur.  By undetected error event, we refer to an \emph{incorrectly} decoded component word with
all-zero syndrome.

In general, for a non-minimal stall pattern of $\bl{i}^{(z)}$ with masking matrices $\mask{i}$ and $\mask{i+1}$ as
in~\defref{def:syndvec}, all positions involved in the stall pattern are covered, as
$\wt(\mask{i}) + \wt(\mask{i+1}) = K\cdot L$.  This mask therefore reconstructs a stall pattern of correct size, but of
maximum weight $K\cdot L$ (compare \eqref{eq:epsbound}) which might introduce $\bar{\epsilon}$ new errors after the
bit-flipping, where
\begin{equation} \label{eq:epsbarbound} \bar{\epsilon} = K \cdot L - \epsilon.
\end{equation}
For example, for the $(4,4)$ non-minimal stall pattern of~\fig{fig:resnonminstpa}, the bit-flip operation resolves the
12 erroneous bits of the stall pattern, but introduces 4 new errors, as indicated by red markers. These 4 new errors can
then be corrected by a usual sliding-window decoding iteration, since the weight in each column and row is less than
$t=2$.
\begin{figure}
  \centering

	\begin{minipage}{0.35\textwidth}
          \def\x{1.5} 

\tikzset{cross/.style={cross out, draw=black, fill=none, minimum size=2*(#1-\pgflinewidth), inner sep=0pt, outer sep=0pt}, cross/.default={\x*2pt}}

\begin{tikzpicture}

	\node[draw=none] at (\x*1.5,\x*2.5)  {\color{light-gray}$\bl{i-1}^T$};
	\node[draw=none] at (\x*2.5,\x*2.5)  {\color{light-gray}$\bl{i}$};
	\node[draw=none] at (\x*2.5,\x*1.5)  {\color{light-gray}$\bl{i+1}^T$};
	\node[draw=none] at (\x*3.5,\x*1.5)  {\color{light-gray}$\bl{i+2}$};

	\draw (\x*1,\x*3) rectangle (\x*2,\x*2);
	\draw (\x*2,\x*3) rectangle (\x*3,\x*2);
	\draw (\x*2,\x*2) rectangle (\x*3,\x*1);
	\draw (\x*3,\x*2) rectangle (\x*4,\x*1);
	
	\draw (\x*1,\x*3) -- (\x*1,\x*3.3);
	\draw (\x*2,\x*3) -- (\x*2,\x*3.3);	
	\draw (\x*1.5,\x*3.1) node[circle,fill,inner sep=\x*0.2pt]{};
	\draw (\x*1.5,\x*3.165) node[circle,fill,inner sep=\x*0.2pt]{};
	\draw (\x*1.5,\x*3.23) node[circle,fill,inner sep=\x*0.2pt]{};
	
	\draw (\x*3,\x*1) -- (\x*3,\x*0.7);
	\draw (\x*4,\x*1) -- (\x*4,\x*0.7);	
	\draw (\x*3.5,\x*0.9) node[circle,fill,inner sep=\x*0.2pt]{};
	\draw (\x*3.5,\x*0.835) node[circle,fill,inner sep=\x*0.2pt]{};
	\draw (\x*3.5,\x*0.77) node[circle,fill,inner sep=\x*0.2pt]{};

	\draw[dashed] (\x*1,\x*2.3) -- (\x*3,\x*2.3);
	\draw[dashed] (\x*1,\x*2.75) -- (\x*3,\x*2.75);
	\draw[dashed] (\x*2,\x*1.2) -- (\x*4,\x*1.2);
	\draw[dashed] (\x*2,\x*1.8) -- (\x*4,\x*1.8);
	
	\draw[dashed] (\x*2.1,\x*3) -- (\x*2.1,\x*1);
	\draw[dashed] (\x*2.3,\x*3) -- (\x*2.3,\x*1);
	\draw[dashed] (\x*2.65,\x*3) -- (\x*2.65,\x*1);
	\draw[dashed] (\x*2.9,\x*3) -- (\x*2.9,\x*1);
	
	\draw (\x*2.1,\x*1.8) node[cross]{};
	\draw (\x*2.1,\x*2.3) node[cross]{};
	\draw (\x*2.1,\x*2.75) node[cross]{};
	\draw (\x*2.3,\x*1.2) node[cross]{};
	\draw (\x*2.3,\x*1.8) node[cross]{};
	\draw (\x*2.3,\x*2.75) node[cross]{};
	\draw (\x*2.65,\x*1.2) node[cross]{};
	\draw (\x*2.65,\x*1.8) node[cross]{};
	\draw (\x*2.65,\x*2.3) node[cross]{};
	\draw (\x*2.9,\x*1.2) node[cross]{};
	\draw (\x*2.9,\x*2.3) node[cross]{};
	\draw (\x*2.9,\x*2.75) node[cross]{};

\end{tikzpicture}
	\end{minipage}
	\hspace{0.05\textwidth} $\stackrel{\text{bit-flip}}{\Rightarrow}$ \hspace{0.03\textwidth}
	\begin{minipage}{0.35\textwidth}
          \def\x{1.5} 

\tikzset{cross/.style={cross out, draw=red, fill=none, minimum size=2*(#1-\pgflinewidth), inner sep=0pt, outer sep=0pt}, cross/.default={\x*2pt}}

\begin{tikzpicture}

	\node[draw=none] at (\x*1.5,\x*2.5)  {\color{light-gray}$\bl{i-1}^T$};
	\node[draw=none] at (\x*2.5,\x*2.5)  {\color{light-gray}$\bl{i}$};
	\node[draw=none] at (\x*2.5,\x*1.5)  {\color{light-gray}$\bl{i+1}^T$};
	\node[draw=none] at (\x*3.5,\x*1.5)  {\color{light-gray}$\bl{i+2}$};

	\draw (\x*1,\x*3) rectangle (\x*2,\x*2);
	\draw (\x*2,\x*3) rectangle (\x*3,\x*2);
	\draw (\x*2,\x*2) rectangle (\x*3,\x*1);
	\draw (\x*3,\x*2) rectangle (\x*4,\x*1);
	
	\draw (\x*1,\x*3) -- (\x*1,\x*3.3);
	\draw (\x*2,\x*3) -- (\x*2,\x*3.3);	
	\draw (\x*1.5,\x*3.1) node[circle,fill,inner sep=\x*0.2pt]{};
	\draw (\x*1.5,\x*3.165) node[circle,fill,inner sep=\x*0.2pt]{};
	\draw (\x*1.5,\x*3.23) node[circle,fill,inner sep=\x*0.2pt]{};
	
	\draw (\x*3,\x*1) -- (\x*3,\x*0.7);
	\draw (\x*4,\x*1) -- (\x*4,\x*0.7);	
	\draw (\x*3.5,\x*0.9) node[circle,fill,inner sep=\x*0.2pt]{};
	\draw (\x*3.5,\x*0.835) node[circle,fill,inner sep=\x*0.2pt]{};
	\draw (\x*3.5,\x*0.77) node[circle,fill,inner sep=\x*0.2pt]{};

	\draw[dashed] (\x*1,\x*2.3) -- (\x*3,\x*2.3);
	\draw[dashed] (\x*1,\x*2.75) -- (\x*3,\x*2.75);
	\draw[dashed] (\x*2,\x*1.2) -- (\x*4,\x*1.2);
	\draw[dashed] (\x*2,\x*1.8) -- (\x*4,\x*1.8);
	
	\draw[dashed] (\x*2.1,\x*3) -- (\x*2.1,\x*1);
	\draw[dashed] (\x*2.3,\x*3) -- (\x*2.3,\x*1);
	\draw[dashed] (\x*2.65,\x*3) -- (\x*2.65,\x*1);
	\draw[dashed] (\x*2.9,\x*3) -- (\x*2.9,\x*1);
	
	\draw (\x*2.1,\x*1.2) node[cross]{};
	\draw (\x*2.3,\x*2.3) node[cross]{};
	\draw (\x*2.65,\x*2.75) node[cross]{};
	\draw (\x*2.9,\x*1.8) node[cross]{};

\end{tikzpicture}
	\end{minipage}
	\caption{Bit-flip operation applied to a non-minimal $(4,4)$ stall pattern of a staircase code with $t=2$
          error-correcting component code. Since the conditions $K < 2(t+1)$ and $L<2(t+1)$ hold, the errors inserted by
          the bit-flip operation can be resolved by decoding of the component codes.}

	\label{fig:resnonminstpa}
      \end{figure}
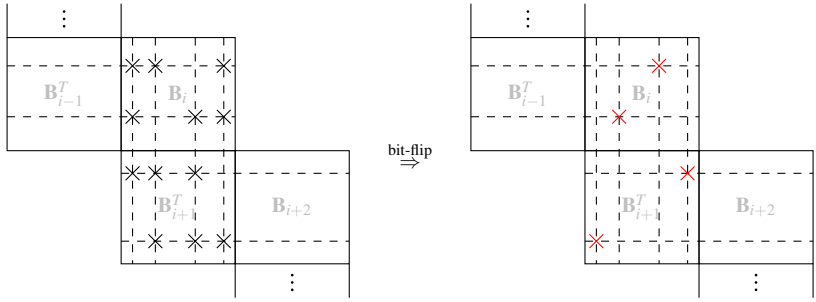

\begin{theorem}[Guaranteed Resolving of Stall Patterns]\label{theo:guarantee-solving}
  Consider a staircase code with an extended BCH component code of minimum distance $d_{\min} = 2t+2$.  Assume that the
  sliding-window decoder has corrected all errors except for stall patterns with $K, L < 2t+2$.  Then, the
  \textbf{bit-flip} operation from Definition~$\ref{def:syndvec}$ and a single normal sliding-window iteration correct
  all these stall patterns if no undetected error events occur.
\end{theorem}
\begin{proof}
  When $K, L < 2(t+1)$, the weight of every row $\rv{}$ of the stall pattern is bounded by $t+1 \leq \wt(\rv{}) \leq L$,
  where the lower bound is given by the definition of stall patterns. When the $L$ involved bits of each row are
  flipped, its weight is bounded by
  \begin{equation} \label{eq:remerrs} \wt(\bar{\rv{}}) = L - \wt(\rv{}) \leq (2(t+1)-1) - (t+1) \leq t ,
  \end{equation}
  which can be corrected by the component codes in a normal sliding-window iteration.
\end{proof}
The restrictions on $K$ and $L$ imply that the error weight in each row or column is at most $d_{\min}-1$ and it follows that erasure decoding could be applied by treating every involved column (row) as an erasure. By guaranteeing the resolving for these restrictions, \theoref{theo:guarantee-solving} shows that bit-flipping offers at least the same performance in terms of stall pattern resolving capability as an approach based on erasure decoding.

For larger $K$ and $L$, the restriction of \eqref{eq:remerrs} no longer holds in general. \fig{fig:nonmininvert} depicts
a stall pattern for which the application of the bit-flip operation from \defref{def:syndvec} leads to another stall
pattern.
\begin{figure}
  \centering

  \begin{minipage}{0.35\textwidth}
    \def\x{1.5} 

\tikzset{cross/.style={cross out, draw=black, fill=none, minimum size=2*(#1-\pgflinewidth), inner sep=0pt, outer sep=0pt}, cross/.default={\x*2pt}}

\begin{tikzpicture}

	\node[draw=none] at (\x*1.5,\x*2.5)  {\color{light-gray}$\bl{i-1}^T$};
	\node[draw=none] at (\x*2.5,\x*2.5)  {\color{light-gray}$\bl{i}$};
	\node[draw=none] at (\x*2.5,\x*1.5)  {\color{light-gray}$\bl{i+1}^T$};
	\node[draw=none] at (\x*3.5,\x*1.5)  {\color{light-gray}$\bl{i+2}$};

	\draw (\x*1,\x*3) rectangle (\x*2,\x*2);
	\draw (\x*2,\x*3) rectangle (\x*3,\x*2);
	\draw (\x*2,\x*2) rectangle (\x*3,\x*1);
	\draw (\x*3,\x*2) rectangle (\x*4,\x*1);
	
	\draw (\x*1,\x*3) -- (\x*1,\x*3.3);
	\draw (\x*2,\x*3) -- (\x*2,\x*3.3);	
	\draw (\x*1.5,\x*3.1) node[circle,fill,inner sep=\x*0.2pt]{};
	\draw (\x*1.5,\x*3.165) node[circle,fill,inner sep=\x*0.2pt]{};
	\draw (\x*1.5,\x*3.23) node[circle,fill,inner sep=\x*0.2pt]{};
	
	\draw (\x*3,\x*1) -- (\x*3,\x*0.7);
	\draw (\x*4,\x*1) -- (\x*4,\x*0.7);	
	\draw (\x*3.5,\x*0.9) node[circle,fill,inner sep=\x*0.2pt]{};
	\draw (\x*3.5,\x*0.835) node[circle,fill,inner sep=\x*0.2pt]{};
	\draw (\x*3.5,\x*0.77) node[circle,fill,inner sep=\x*0.2pt]{};

	\draw[dashed] (\x*1,\x*2.1) -- (\x*3,\x*2.1);
	\draw[dashed] (\x*1,\x*2.3) -- (\x*3,\x*2.3);
	\draw[dashed] (\x*1,\x*2.75) -- (\x*3,\x*2.75);
	\draw[dashed] (\x*2,\x*1.2) -- (\x*4,\x*1.2);
	\draw[dashed] (\x*2,\x*1.8) -- (\x*4,\x*1.8);
	\draw[dashed] (\x*2,\x*1.4) -- (\x*4,\x*1.4);
	
	\draw[dashed] (\x*2.1,\x*3) -- (\x*2.1,\x*1);
	\draw[dashed] (\x*2.3,\x*3) -- (\x*2.3,\x*1);
	\draw[dashed] (\x*2.45,\x*3) -- (\x*2.45,\x*1);
	\draw[dashed] (\x*2.6,\x*3) -- (\x*2.6,\x*1);
	\draw[dashed] (\x*2.9,\x*3) -- (\x*2.9,\x*1);
	\draw[dashed] (\x*2.77,\x*3) -- (\x*2.77,\x*1);
		
	\draw (\x*2.1,\x*2.75) node[cross]{};
	\draw (\x*2.3,\x*2.75) node[cross]{};
	\draw (\x*2.45,\x*2.75) node[cross]{};
	
	\draw (\x*2.3,\x*2.3) node[cross]{};
	\draw (\x*2.45,\x*2.3) node[cross]{};
	\draw (\x*2.6,\x*2.3) node[cross]{};
	
	\draw (\x*2.45,\x*2.1) node[cross]{};
	\draw (\x*2.6,\x*2.1) node[cross]{};
	\draw (\x*2.77,\x*2.1) node[cross]{};

	\draw (\x*2.6,\x*1.8) node[cross]{};
	\draw (\x*2.77,\x*1.8) node[cross]{};
	\draw (\x*2.9,\x*1.8) node[cross]{};
	
	\draw (\x*2.77,\x*1.4) node[cross]{};
	\draw (\x*2.9,\x*1.4) node[cross]{};
	\draw (\x*2.1,\x*1.4) node[cross]{};
	
	\draw (\x*2.9,\x*1.2) node[cross]{};
	\draw (\x*2.1,\x*1.2) node[cross]{};
	\draw (\x*2.3,\x*1.2) node[cross]{};

	\node[draw=none] at (\x*2,\x*0.6) {};
\end{tikzpicture}
  \end{minipage}
  \hspace{0.05\textwidth} $\stackrel{\text{bit-flip}}{\Rightarrow}$ \hspace{0.03\textwidth}
  \begin{minipage}{0.35\textwidth}
    \def\x{1.5} 

\tikzset{cross/.style={cross out, draw=red, fill=none, minimum size=2*(#1-\pgflinewidth), inner sep=0pt, outer sep=0pt}, cross/.default={\x*2pt}}

\begin{tikzpicture}

	\node[draw=none] at (\x*1.5,\x*2.5)  {\color{light-gray}$\bl{i-1}^T$};
	\node[draw=none] at (\x*2.5,\x*2.5)  {\color{light-gray}$\bl{i}$};
	\node[draw=none] at (\x*2.5,\x*1.5)  {\color{light-gray}$\bl{i+1}^T$};
	\node[draw=none] at (\x*3.5,\x*1.5)  {\color{light-gray}$\bl{i+2}$};

	\draw (\x*1,\x*3) rectangle (\x*2,\x*2);
	\draw (\x*2,\x*3) rectangle (\x*3,\x*2);
	\draw (\x*2,\x*2) rectangle (\x*3,\x*1);
	\draw (\x*3,\x*2) rectangle (\x*4,\x*1);
	
	\draw (\x*1,\x*3) -- (\x*1,\x*3.3);
	\draw (\x*2,\x*3) -- (\x*2,\x*3.3);	
	\draw (\x*1.5,\x*3.1) node[circle,fill,inner sep=\x*0.2pt]{};
	\draw (\x*1.5,\x*3.165) node[circle,fill,inner sep=\x*0.2pt]{};
	\draw (\x*1.5,\x*3.23) node[circle,fill,inner sep=\x*0.2pt]{};
	
	\draw (\x*3,\x*1) -- (\x*3,\x*0.7);
	\draw (\x*4,\x*1) -- (\x*4,\x*0.7);	
	\draw (\x*3.5,\x*0.9) node[circle,fill,inner sep=\x*0.2pt]{};
	\draw (\x*3.5,\x*0.835) node[circle,fill,inner sep=\x*0.2pt]{};
	\draw (\x*3.5,\x*0.77) node[circle,fill,inner sep=\x*0.2pt]{};

	\draw[dashed] (\x*1,\x*2.1) -- (\x*3,\x*2.1);
	\draw[dashed] (\x*1,\x*2.3) -- (\x*3,\x*2.3);
	\draw[dashed] (\x*1,\x*2.75) -- (\x*3,\x*2.75);
	\draw[dashed] (\x*2,\x*1.2) -- (\x*4,\x*1.2);
	\draw[dashed] (\x*2,\x*1.8) -- (\x*4,\x*1.8);
	\draw[dashed] (\x*2,\x*1.4) -- (\x*4,\x*1.4);
	
	\draw[dashed] (\x*2.1,\x*3) -- (\x*2.1,\x*1);
	\draw[dashed] (\x*2.3,\x*3) -- (\x*2.3,\x*1);
	\draw[dashed] (\x*2.45,\x*3) -- (\x*2.45,\x*1);
	\draw[dashed] (\x*2.6,\x*3) -- (\x*2.6,\x*1);
	\draw[dashed] (\x*2.9,\x*3) -- (\x*2.9,\x*1);
	\draw[dashed] (\x*2.77,\x*3) -- (\x*2.77,\x*1);
	
	\draw (\x*2.6,\x*2.75) node[cross]{};
	\draw (\x*2.77,\x*2.75) node[cross]{};
	\draw (\x*2.9,\x*2.75) node[cross]{};
	
	\draw (\x*2.77,\x*2.3) node[cross]{};
	\draw (\x*2.9,\x*2.3) node[cross]{};
	\draw (\x*2.1,\x*2.3) node[cross]{};
	
	\draw (\x*2.9,\x*2.1) node[cross]{};
	\draw (\x*2.1,\x*2.1) node[cross]{};
	\draw (\x*2.3,\x*2.1) node[cross]{};
	
	\draw (\x*2.1,\x*1.8) node[cross]{};
	\draw (\x*2.3,\x*1.8) node[cross]{};
	\draw (\x*2.45,\x*1.8) node[cross]{};
	
	\draw (\x*2.3,\x*1.4) node[cross]{};
	\draw (\x*2.45,\x*1.4) node[cross]{};
	\draw (\x*2.6,\x*1.4) node[cross]{};
	
	\draw (\x*2.45,\x*1.2) node[cross]{};
	\draw (\x*2.6,\x*1.2) node[cross]{};
	\draw (\x*2.77,\x*1.2) node[cross]{};

	\node[draw=none] at (\x*2,\x*0.6) {};
\end{tikzpicture}
  \end{minipage}
  \caption{Non-minimal $(6,6)$ stall pattern of a staircase code with $t=2$ error-correcting component code. Applying
    the bit-flip operation results in another stall pattern of same size.}

  \label{fig:nonmininvert}
\end{figure}
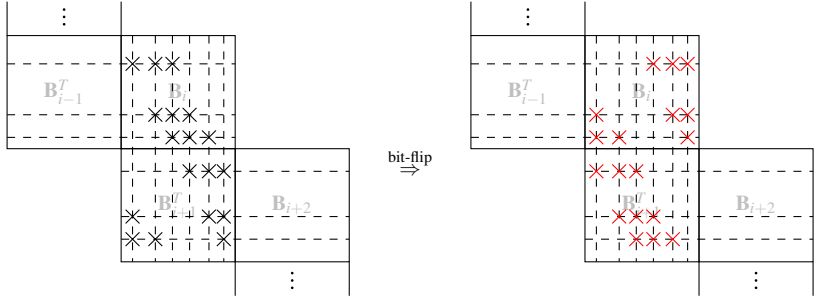

\subsection{Bit-Flip with Undetected Error Events} \label{subsec:resundet}

Assume that undetected error events occur, \emph{i.e.}, there is an incorrect component word with all-zero syndrome after the
sliding window decoding. This word is a codeword of the component code, but since every positions is protected by two
component codes, the errors can generally still be detected by the other component code decoder. However, if not only
one but multiple undetected error events occur such that the resulting errors are in the same positions, it is possible
that not all positions involved in a stall pattern can be located and~\eqref{eq:remerrs} does not necessarily hold. It
is therefore difficult to give a theoretical analysis of stall patterns with undetected error events, but simulations
(Section~\ref{subsec:simres}) show that these cases are unlikely and that many such patterns can still be
resolved. \fig{fig:43stpa} shows an example of an uncorrectable $(4,3)$ stall pattern with three identical error vectors
in the columns and additional undetected erroneous rows that cannot be resolved at all. While the decoder detects that
the block is not valid, the columns cannot be located because their syndromes are zero and the operation \emph{bit-flip}
fails. In this case, the errors of the stall pattern are detectable but not correctable.

\begin{figure}
  \centering

  \begin{minipage}{0.35\textwidth}
    \def\x{1.5} 

\tikzset{cross/.style={cross out, draw=black, fill=none, minimum size=2*(#1-\pgflinewidth), inner sep=0pt, outer sep=0pt}, cross/.default={\x*2pt}}

\begin{tikzpicture}

	\node[draw=none] at (\x*1.5,\x*2.5)  {\color{light-gray}$\bl{i-1}^T$};
	\node[draw=none] at (\x*2.5,\x*2.5)  {\color{light-gray}$\bl{i}$};
	\node[draw=none] at (\x*2.5,\x*1.5)  {\color{light-gray}$\bl{i+1}^T$};
	\node[draw=none] at (\x*3.5,\x*1.5)  {\color{light-gray}$\bl{i+2}$};

	\draw (\x*1,\x*3) rectangle (\x*2,\x*2);
	\draw (\x*2,\x*3) rectangle (\x*3,\x*2);
	\draw (\x*2,\x*2) rectangle (\x*3,\x*1);
	\draw (\x*3,\x*2) rectangle (\x*4,\x*1);
	
	\draw (\x*1,\x*3) -- (\x*1,\x*3.3);
	\draw (\x*2,\x*3) -- (\x*2,\x*3.3);	
	\draw (\x*1.5,\x*3.1) node[circle,fill,inner sep=\x*0.2pt]{};
	\draw (\x*1.5,\x*3.165) node[circle,fill,inner sep=\x*0.2pt]{};
	\draw (\x*1.5,\x*3.23) node[circle,fill,inner sep=\x*0.2pt]{};

	\draw (\x*3,\x*1) -- (\x*3,\x*0.7);
	\draw (\x*4,\x*1) -- (\x*4,\x*0.7);	
	\draw (\x*3.5,\x*0.9) node[circle,fill,inner sep=\x*0.2pt]{};
	\draw (\x*3.5,\x*0.835) node[circle,fill,inner sep=\x*0.2pt]{};
	\draw (\x*3.5,\x*0.77) node[circle,fill,inner sep=\x*0.2pt]{};

	\draw[dashed] (\x*1,\x*2.3) -- (\x*3,\x*2.3);
	\draw[dashed] (\x*1,\x*2.75) -- (\x*3,\x*2.75);
	\draw[dashed] (\x*2,\x*1.2) -- (\x*4,\x*1.2);
	\draw[dashed] (\x*2,\x*1.8) -- (\x*4,\x*1.8);
	
	\draw[dashed] (\x*2.2,\x*3) -- (\x*2.2,\x*1);
	\draw[dashed] (\x*2.5,\x*3) -- (\x*2.5,\x*1);
	\draw[dashed] (\x*2.8,\x*3) -- (\x*2.8,\x*1);
	
	\draw (\x*2.2,\x*1.2) node[cross]{};
	\draw (\x*2.2,\x*1.8) node[cross]{};
	\draw (\x*2.2,\x*2.3) node[cross]{};
	\draw (\x*2.2,\x*2.75) node[cross]{};
	\draw (\x*2.5,\x*1.2) node[cross]{};
	\draw (\x*2.5,\x*1.8) node[cross]{};
	\draw (\x*2.5,\x*2.3) node[cross]{};
	\draw (\x*2.5,\x*2.75) node[cross]{};
	\draw (\x*2.8,\x*1.2) node[cross]{};
	\draw (\x*2.8,\x*1.8) node[cross]{};
	\draw (\x*2.8,\x*2.3) node[cross]{};
	\draw (\x*2.8,\x*2.75) node[cross]{};

\end{tikzpicture}
  \end{minipage}
  \hspace{0.05\textwidth} $\stackrel{\text{decode}}{\Rightarrow}$ \hspace{0.03\textwidth}
  \begin{minipage}{0.35\textwidth}
    \def\x{1.5} 

\tikzset{cross/.style={cross out, fill=none, minimum size=2*(#1-\pgflinewidth), inner sep=0pt, outer sep=0pt}, cross/.default={\x*2pt}}

\begin{tikzpicture}

	\node[draw=none] at (\x*1.5,\x*2.5)  {\color{light-gray}$\bl{i-1}^T$};
	\node[draw=none] at (\x*2.5,\x*2.5)  {\color{light-gray}$\bl{i}$};
	\node[draw=none] at (\x*2.5,\x*1.5)  {\color{light-gray}$\bl{i+1}^T$};
	\node[draw=none] at (\x*3.5,\x*1.5)  {\color{light-gray}$\bl{i+2}$};

	\draw (\x*1,\x*3) rectangle (\x*2,\x*2);
	\draw (\x*2,\x*3) rectangle (\x*3,\x*2);
	\draw (\x*2,\x*2) rectangle (\x*3,\x*1);
	\draw (\x*3,\x*2) rectangle (\x*4,\x*1);
	
	\draw (\x*1,\x*3) -- (\x*1,\x*3.3);
	\draw (\x*2,\x*3) -- (\x*2,\x*3.3);	
	\draw (\x*1.5,\x*3.1) node[circle,fill,inner sep=\x*0.2pt]{};
	\draw (\x*1.5,\x*3.165) node[circle,fill,inner sep=\x*0.2pt]{};
	\draw (\x*1.5,\x*3.23) node[circle,fill,inner sep=\x*0.2pt]{};

	\draw (\x*3,\x*1) -- (\x*3,\x*0.7);
	\draw (\x*4,\x*1) -- (\x*4,\x*0.7);	
	\draw (\x*3.5,\x*0.9) node[circle,fill,inner sep=\x*0.2pt]{};
	\draw (\x*3.5,\x*0.835) node[circle,fill,inner sep=\x*0.2pt]{};
	\draw (\x*3.5,\x*0.77) node[circle,fill,inner sep=\x*0.2pt]{};

	\draw[dashed] (\x*1,\x*2.3) -- (\x*3,\x*2.3);
	\draw[dashed] (\x*1,\x*2.75) -- (\x*3,\x*2.75);
	\draw[dashed] (\x*2,\x*1.2) -- (\x*4,\x*1.2);
	\draw[dashed] (\x*2,\x*1.8) -- (\x*4,\x*1.8);
	
	\draw[dashed,blue] (\x*1,\x*2.5) -- (\x*3,\x*2.5);
	\draw[dashed,blue] (\x*2,\x*1.5) -- (\x*4,\x*1.5);
		
	\draw (\x*2.2,\x*1.2) node[cross,draw=black]{};
	\draw (\x*2.2,\x*1.8) node[cross,draw=black]{};
	\draw (\x*2.2,\x*2.3) node[cross,draw=black]{};
	\draw (\x*2.2,\x*2.75) node[cross,draw=black]{};
	\draw (\x*2.5,\x*1.2) node[cross,draw=black]{};
	\draw (\x*2.5,\x*1.8) node[cross,draw=black]{};
	\draw (\x*2.5,\x*2.3) node[cross,draw=black]{};
	\draw (\x*2.5,\x*2.75) node[cross,draw=black]{};
	\draw (\x*2.8,\x*1.2) node[cross,draw=black]{};
	\draw (\x*2.8,\x*1.8) node[cross,draw=black]{};
	\draw (\x*2.8,\x*2.3) node[cross,draw=black]{};
	\draw (\x*2.8,\x*2.75) node[cross,draw=black]{};
	
	\draw (\x*2.2,\x*2.5) node[cross,draw=blue]{};
	\draw (\x*2.5,\x*2.5) node[cross,draw=blue]{};
	\draw (\x*2.8,\x*2.5) node[cross,draw=blue]{};
	\draw (\x*2.2,\x*1.5) node[cross,draw=blue]{};
	\draw (\x*2.5,\x*1.5) node[cross,draw=blue]{};
	\draw (\x*2.8,\x*1.5) node[cross,draw=blue]{};

\end{tikzpicture}
  \end{minipage}
  \caption{Non-minimal $(4,3)$ stall pattern of a staircase code with $t=2$ error-correcting component codes of distance
    $d_{\min}=6$. The stall pattern cannot be resolved because the errors inserted by undetected error events are not
    correctable and the columns cannot be located.}
  \label{fig:43stpa}
\end{figure}
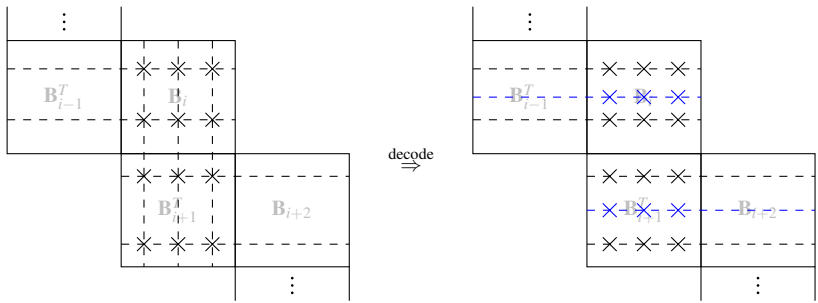

The difficulties of decoding the previously discussed unsolvable stall pattern stem from not being able to locate the
involved rows or columns at all. However, if the involved rows and columns both cause the error vectors to be in the
proximity of a valid codeword, undetected error events in the component codewords can cause the error matrix to be close
to a valid staircase code block. As the component codes are linear, so is the staircase code itself and an error matrix
that is a valid block is therefore not detectable by any decoding algorithm. \fig{fig:4414stpa} gives an illustration of
such an unsolvable stall pattern. The error matrix differs, depending on which decoder (row or column) ran last. It is
possible that such a pattern is resolved by inverting only a single column and then performing regular decoding
iterations, as described in \sect{subsec:resnoundet}. However, it now depends on the chosen column. If more errors are
inserted than resolved, the weight of the stall pattern will increase due to undetected error events in all rows and
columns, resulting in an undetected error pattern in the staircase code.

\begin{figure}
  \centering
  \begin{minipage}{0.35\textwidth}
    \def\x{1.5} 

\tikzset{cross/.style={cross out, fill=none, minimum size=2*(#1-\pgflinewidth), inner sep=0pt, outer sep=0pt}, cross/.default={\x*2pt}}

\begin{tikzpicture}

	\node[draw=none] at (\x*1.5,\x*2.5)  {\color{light-gray}$\bl{i-1}^T$};
	\node[draw=none] at (\x*2.5,\x*2.5)  {\color{light-gray}$\bl{i}$};
	\node[draw=none] at (\x*2.5,\x*1.5)  {\color{light-gray}$\bl{i+1}^T$};
	\node[draw=none] at (\x*3.5,\x*1.5)  {\color{light-gray}$\bl{i+2}$};

	\draw (\x*1,\x*3) rectangle (\x*2,\x*2);
	\draw (\x*2,\x*3) rectangle (\x*3,\x*2);
	\draw (\x*2,\x*2) rectangle (\x*3,\x*1);
	\draw (\x*3,\x*2) rectangle (\x*4,\x*1);
	
	\draw (\x*1,\x*3) -- (\x*1,\x*3.3);
	\draw (\x*2,\x*3) -- (\x*2,\x*3.3);	
	\draw (\x*1.5,\x*3.1) node[circle,fill,inner sep=\x*0.2pt]{};
	\draw (\x*1.5,\x*3.165) node[circle,fill,inner sep=\x*0.2pt]{};
	\draw (\x*1.5,\x*3.23) node[circle,fill,inner sep=\x*0.2pt]{};
	
	\draw (\x*3,\x*1) -- (\x*3,\x*0.7);
	\draw (\x*4,\x*1) -- (\x*4,\x*0.7);	
	\draw (\x*3.5,\x*0.9) node[circle,fill,inner sep=\x*0.2pt]{};
	\draw (\x*3.5,\x*0.835) node[circle,fill,inner sep=\x*0.2pt]{};
	\draw (\x*3.5,\x*0.77) node[circle,fill,inner sep=\x*0.2pt]{};

	\draw[dashed] (\x*1,\x*2.2) -- (\x*3,\x*2.2);
	\draw[dashed] (\x*2,\x*1.8) -- (\x*4,\x*1.8);

	\draw[dashed] (\x*2.7,\x*3) -- (\x*2.7,\x*1);
	\draw[dashed] (\x*2.25,\x*3) -- (\x*2.25,\x*1);
	\draw[dashed] (\x*2.4,\x*3) -- (\x*2.4,\x*1);
	\draw[dashed] (\x*2.55,\x*3) -- (\x*2.55,\x*1);
	\draw[dashed,blue] (\x*2.1,\x*3) -- (\x*2.1,\x*1);
	\draw[dashed,blue] (\x*2.9,\x*3) -- (\x*2.9,\x*1);
	
	\draw (\x*2.7,\x*1.3) node[cross,draw=black]{};
	\draw (\x*2.7,\x*1.8) node[cross,draw=black]{};
	\draw (\x*2.7,\x*2.75) node[cross,draw=black]{};
	\draw (\x*2.7,\x*2.2) node[cross,draw=black]{};
	\draw (\x*2.25,\x*1.3) node[cross,draw=black]{};
	\draw (\x*2.25,\x*1.8) node[cross,draw=black]{};
	\draw (\x*2.25,\x*2.75) node[cross,draw=black]{};
	\draw (\x*2.4,\x*1.8) node[cross,draw=black]{};
	\draw (\x*2.4,\x*1.3) node[cross,draw=black]{};
	\draw (\x*2.4,\x*2.2) node[cross,draw=black]{};
	\draw (\x*2.4,\x*2.75) node[cross,draw=black]{};
	\draw (\x*2.55,\x*1.3) node[cross,draw=black]{};
	\draw (\x*2.55,\x*2.75) node[cross,draw=black]{};
	\draw (\x*2.55,\x*2.2) node[cross,draw=black]{};
	\draw (\x*2.1,\x*1.3) node[cross,draw=blue]{};
	\draw (\x*2.1,\x*2.75) node[cross,draw=blue]{};
	\draw (\x*2.9,\x*1.3) node[cross,draw=blue]{};
	\draw (\x*2.9,\x*2.75) node[cross,draw=blue]{};
	
\end{tikzpicture}
  \end{minipage}
  \hspace{0.05\textwidth}$\mathop{\rightleftarrows}\limits^{\text{dec. columns}}_{\text{dec. rows}}$
  \hspace{0.003\textwidth}
  \begin{minipage}{0.35\textwidth}
    \def\x{1.5} 

\tikzset{cross/.style={cross out, fill=none, minimum size=2*(#1-\pgflinewidth), inner sep=0pt, outer sep=0pt}, cross/.default={\x*2pt}}

\begin{tikzpicture}

	\node[draw=none] at (\x*1.5,\x*2.5)  {\color{light-gray}$\bl{i-1}^T$};
	\node[draw=none] at (\x*2.5,\x*2.5)  {\color{light-gray}$\bl{i}$};
	\node[draw=none] at (\x*2.5,\x*1.5)  {\color{light-gray}$\bl{i+1}^T$};
	\node[draw=none] at (\x*3.5,\x*1.5)  {\color{light-gray}$\bl{i+2}$};

	\draw (\x*1,\x*3) rectangle (\x*2,\x*2);
	\draw (\x*2,\x*3) rectangle (\x*3,\x*2);
	\draw (\x*2,\x*2) rectangle (\x*3,\x*1);
	\draw (\x*3,\x*2) rectangle (\x*4,\x*1);
	
	\draw (\x*1,\x*3) -- (\x*1,\x*3.3);
	\draw (\x*2,\x*3) -- (\x*2,\x*3.3);	
	\draw (\x*1.5,\x*3.1) node[circle,fill,inner sep=\x*0.2pt]{};
	\draw (\x*1.5,\x*3.165) node[circle,fill,inner sep=\x*0.2pt]{};
	\draw (\x*1.5,\x*3.23) node[circle,fill,inner sep=\x*0.2pt]{};
	
	\draw (\x*3,\x*1) -- (\x*3,\x*0.7);
	\draw (\x*4,\x*1) -- (\x*4,\x*0.7);	
	\draw (\x*3.5,\x*0.9) node[circle,fill,inner sep=\x*0.2pt]{};
	\draw (\x*3.5,\x*0.835) node[circle,fill,inner sep=\x*0.2pt]{};
	\draw (\x*3.5,\x*0.77) node[circle,fill,inner sep=\x*0.2pt]{};

	\draw[dashed] (\x*1,\x*2.2) -- (\x*3,\x*2.2);
	\draw[dashed] (\x*1,\x*2.75) -- (\x*3,\x*2.75);
	\draw[dashed] (\x*2,\x*1.3) -- (\x*4,\x*1.3);
	\draw[dashed] (\x*2,\x*1.8) -- (\x*4,\x*1.8);
	\draw[dashed,blue] (\x*2,\x*1.5) -- (\x*4,\x*1.5);
	\draw[dashed,blue] (\x*1,\x*2.4) -- (\x*3,\x*2.4);
	
	\draw[dashed] (\x*2.25,\x*3) -- (\x*2.25,\x*1);
	\draw[dashed] (\x*2.55,\x*3) -- (\x*2.55,\x*1);
	
	\draw (\x*2.7,\x*1.3) node[cross,draw=black]{};
	\draw (\x*2.7,\x*1.8) node[cross,draw=black]{};
	\draw (\x*2.7,\x*2.75) node[cross,draw=black]{};
	\draw (\x*2.7,\x*2.2) node[cross,draw=black]{};
	\draw (\x*2.25,\x*1.3) node[cross,draw=black]{};
	\draw (\x*2.25,\x*1.8) node[cross,draw=black]{};
	\draw (\x*2.25,\x*2.75) node[cross,draw=black]{};
	\draw (\x*2.4,\x*1.8) node[cross,draw=black]{};
	\draw (\x*2.4,\x*1.3) node[cross,draw=black]{};
	\draw (\x*2.4,\x*2.2) node[cross,draw=black]{};
	\draw (\x*2.4,\x*2.75) node[cross,draw=black]{};
	\draw (\x*2.55,\x*1.3) node[cross,draw=black]{};
	\draw (\x*2.55,\x*2.75) node[cross,draw=black]{};
	\draw (\x*2.55,\x*2.2) node[cross,draw=black]{};

	\draw (\x*2.7,\x*1.5) node[cross,draw=blue]{};
	\draw (\x*2.7,\x*2.4) node[cross,draw=blue]{};
	\draw (\x*2.4,\x*1.5) node[cross,draw=blue]{};
	\draw (\x*2.4,\x*2.4) node[cross,draw=blue]{};
	
\end{tikzpicture}
  \end{minipage}
  \caption{Non-minimal $(4,4)$ stall pattern of a staircase code with $t=2$ error-correcting component codes of distance
    $d_{\min}=6$. Undetected error events occur in the row decoder (left) and the column decoder (right).}
  \label{fig:4414stpa}
\end{figure}
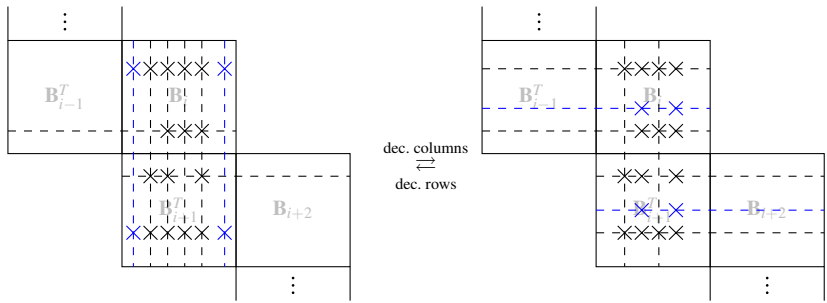

\subsection{The Bit-Flip Algorithm in Implementation} \label{subsec:bitflipimplement}

\sect{subsec:resnoundet} and \sect{subsec:resundet} introduce multiple types of stall patterns, some of which can be
guaranteed to be resolved, while for others it depends on the positions where the errors occurred. \defref{def:syndvec}
offers a mathematical description of the bit-flip operation, which is useful for the analysis. However, for
implementations it makes sense to perform some additional steps, which help prevent undesirable effects of undetected
error events when decoding non-minimal stall patterns. Clearly, if $\epsilon \geq d_{\code{}} = d_{\min}^2$ it is
possible that a stall pattern is undetectable. However, this case is unlikely and it is theoretically impossible to
protect against it. The more likely problem arises when $K \geq d_{\min}$ and $L \geq d_{\min}$, but
$\epsilon < d_{\code{}}$. Then the $\bar{\epsilon}$ previously error free positions can form a stall pattern of size
$(K,L)$ causing the bit-flip operation to result in another stall pattern (compare \fig{fig:nonmininvert}). Further, it
is possible that more errors are inserted in the iterations of the sliding window decoder following the bit-flip
operation. To avoid these undesired effects, the bit-flip operation should be adapted if $K \geq d_{\min}$ and
$L \geq d_{\min}$. In this case, it is advantageous to only flip the bits in a single row or column and then perform
some iterations of the sliding window decoder in which only errors within the positions involved in the stall pattern
are decoded, before performing the regular iterations again. Then, with high probability, some rows or columns involved
in the stall pattern can be decoded, leading to other decodable columns or rows and eventually to resolving of the stall
pattern, while the risk of inserting more errors outside the positions of the stall pattern is reduced
significantly. For example, consider the stall pattern depicted in \fig{fig:nonmininvert}. Flipping all positions
involved in the stall pattern results in another stall pattern of same size. If only the positions of the first row are
flipped, the three rightmost columns can be decoded in normal sliding window iterations. However, the three leftmost
columns now contain four errors each, making it possible for undetected error events to occur within these columns. When
performing two iterations in which only errors within the positions involved in the stall pattern are corrected, the
possible undetected error events are avoided while the pattern can still be resolved.
\begin{myremark}
  When applying the bit-flip operation, the number of errors inserted is inversely related to the number of errors
  resolved in the respective row or column, \emph{i.e.}, the more errors the row or column had in the positions involved in the
  stall pattern, the fewer it has after the bit-flip operation. It follows that errors that could be decoded by the
  component code, but are still non-zero when the operation is invoked, cause the most errors to be inserted
  additionally. As observed in \cite{Justesen2011} it can advantageous to perform iterations correcting only a single
  error, as it is more likely for an undetected error event to occur when correcting $t$ errors. This resolves
  most errors inserted by undetected error events, while avoiding the insertion of additional errors in most cases and
  can improve the likelihood of resolving a stall pattern (see \stepref{step:repeat} in Algorithm~\ref{algo:slidingresolve}).
\end{myremark}
\begin{algorithm}\label{algo:slidingresolve}
  \caption{Sliding window decoder with stall pattern resolving}
  \KwIn{$N$ Blocks $\bl{0},...,\bl{N-1}$ of a staircase code with $t$ error-correcting $[n,k]$ component code}
  \KwOut{Decoded staircase code blocks}

  $i\leftarrow 0$;\\
  \While{$i + W \leq N$}{

    Perform regular decoding iterations in sliding window until maximum $v_{\max}$ is reached; \label{step:regular}\\

    Perform decoding iteration correcting only one error in every codeword;\label{step:repeat}\\

    \tcc{Get number of rows with non-zero syndromes}
    $\delta_0 \leftarrow$ Number of erroneous rows in $\left[\bl{i\phantom{+1} }^T \;\; \bl{i+1}\right]$; \\
    $\delta_1 \leftarrow$ Number of erroneous rows in $\left[\bl{i+1}^T \;\; \bl{i+2}\right]$;\\
    $\delta_2 \leftarrow$ Number of erroneous rows in $\left[\bl{i+2}^T \;\; \bl{i+3}\right]$;\\

    \If{$\delta_0 \neq 0$}{ \label{step:invoke} \tcc{invoke bit-flip operation} \If{$\delta_0 + \delta_2 < d_{\min}$
        \KWor $\delta_1 < d_{\min}$}{ \label{step:checksmall} \tcc{Stall pattern can likely resolved by flipping all
          positions}
        Flip all positions within $\bl{i+1}$ and $\bl{i+2}$ at intersections of erroneous rows and columns;\\
      } \Else{ \tcc{For large stall patterns flip only one row}
        Flip positions within $\bl{i+1}$ and $\bl{i+2}$ of one involved row;\\
      } \tcc{Attempt resolving the remaining errors}
      Perform decoding iterations only correcting errors within the positions of the stall pattern; \label{step:internal}\\
      Perform decoding iterations only correcting errors in the blocks $\bl{i+1}$ and $\bl{i+2}$; \label{step:normal}\\
      Repeat once from \stepref{step:repeat}; \label{step:jump}\\
    }

    \tcc{Shift the window by one block}

    $i \leftarrow i+1$;\\
  }

\end{algorithm}
\algo{algo:slidingresolve} gives a pseudo-code description of the sliding window decoder with resolving of stall
patterns for a sliding window size $W$ and a maximum number of iterations $v_{\max}$ within one window. An advantage
of staircase codes is that they can be decoded by calculating the syndromes for every component codeword once and then
operating only on those syndromes. For ease of notation, the algorithm is given as operating directly on blocks,
however, it is straight forward implementable in the syndrome domain. The decoder consists of a sliding window decoder
(see \stepref{step:regular}) as proposed in \cite{stairfec} followed by the bit-flipping algorithm proposed in this
work. After the decoding iterations within one window, an additional iteration correcting only a single error is
performed (see \stepref{step:repeat}) to resolve some of the errors inserted by undetected error events, while avoiding
the more likely undetected error events of $t$ inserted errors. In \stepref{step:invoke} it is determined whether the
bit-flip operation should be invoked and in \stepref{step:checksmall} it is determined whether all positions at
intersections of erroneous rows and columns should be flipped (see \theoref{theo:guarantee-solving}) or if it is
advantageous to only flip a single row or column. The bit-flip operation is followed by decoding iterations in which
only positions involved in the stall pattern can be changed (see \stepref{step:internal}), followed by decoding
iterations solving errors in the blocks that contain positions of the stall pattern (see \stepref{step:normal}). As
resolving a large stall pattern can result in other stall patterns, the bit-flip operation should be performed twice
(see \stepref{step:jump}).

In general, the window size of a sliding window decoder should be chosen such that the last block in the sliding window
contains no more decodable errors after the iterations in the window have been performed. Assuming the window size of
the decoder introduced in \cite{stairfec} is chosen in that way, only that block could be used to locate stall
patterns. Since the resolving of stall patterns described in \algo{algo:slidingresolve} requires $4$ blocks to be free
of all errors except for the ones involved in a stall pattern, the size of the sliding window should be increased by $3$
blocks compared to the window size of \cite{stairfec}.

\section{Improved Error-Floor Analysis} \label{sec:impranalysis}

\subsection{An Improved Analysis}
For the parameter ranges of interest, the error floor of the standard sliding-window decoder (see
Section~\ref{subsec:stcdecoding}) is dominated by minimal stall patterns \cite{stairfec}.  As shown in
\sect{subsec:loerrflo}, the improved decoder is able to solve all minimal stall patterns and it follows that an exact
analysis of the error floor contribution of the remaining unsolvable stall patterns is crucial for evaluating its
performance.

As the success of the resolving strategy described in \sect{subsec:loerrflo} depends on $K, L$ and also $\epsilon$, we
analyze each summand from~\eqref{eq:bounderrflo} separately by
\begin{equation} \label{eq:PCold} P_{C,old}(K,L,\epsilon) \stackrel{\triangle}{=} \frac{\epsilon}{m^2} \cdot A_{K,L}
  \cdot \hat{N}_{K,L}^{\epsilon} \cdot (p + \xi)^\epsilon .
\end{equation}
Since minimal stall patterns can be resolved by bit-flipping, their contribution to the error floor is no longer
dominating and using $\hat{N}_{K,L}^\epsilon$ from \eqref{eq:amerrdist} is very inaccurate for the improved decoder, due to
the overestimation of $N_{K,L}^{\epsilon}$ for non-minimal stall patterns~(see \fig{fig:estcomp66}).

The problem of finding the number of stall patterns of weight $\epsilon$ within $K$ rows and $L$ columns is equivalent
to the problem of finding the number of binary $K\times L$ matrices of weight $\epsilon$ and a given weight in each row
and column. A solution to this combinatorial problem is given in \cite{mat01} (see also \cite{Perez2002}). The function
denoted by $\mathcal{A}(\mathbf{r},\mathbf{s})$ takes a vector $\mathbf{r} = \left[r_1,r_2,...,r_{K} \right]$ containing
the row weights and $\mathbf{s} = \left[s_1,s_2,...,s_L\right] $ containing the column weights and returns the number of
distinct binary matrices that meet the weight restrictions on the rows and columns. By definition, for stall patterns it
holds that $r_i \geq t+1 \; \forall \; i \in [1,K]$ and $s_j \geq t+1 \; \forall \; j \in [1,L]$, where we denote by
$[a,b]$ the set of integers $i$ such that $a\leq i \leq b$. Since every error has to be in one of the rows and one of
the columns, $\sum_{i=1}^{K} r_i = \epsilon$ and $\sum_{j=1}^{L} s_j = \epsilon$ hold.
\begin{lemma}\label{lem:numberstpa}
  The number of stall patterns of weight $\epsilon$ involving $K$ rows and $L$ columns is
  \begin{equation*}
    N_{K,L}^{\epsilon} = \left| \mathcal{Z}_{K,L}^\epsilon \right|
  \end{equation*}
  where $\mathcal{Z}_{K,L}^\epsilon$ is given in Lemma~\ref{lem:cardZex} (Appendix).
\end{lemma}
\begin{proof}
  Restricting the block to the intersection of the involved rows and columns, each stall pattern of size $(K,L)$ and
  weight $\epsilon$ can be represented by a binary $K\times L$ matrix of weight~$\epsilon$. By definition
  $\mathcal{Z}_{K,L}^{\epsilon}$ is the set of all such matrices and it follows that the number of stall patterns is
  given by its cardinality.
\end{proof}
With \lem{lem:numberstpa} and \lem{lem:cardZex}, the contribution of stall patterns of a given size to the error floor
can be stated without overestimating the number of unique stall patterns.
\begin{theorem}\label{theo:contribution}
  Consider a staircase code of block size $m \times m$. The contribution to the error floor by stall patterns of size
  $(K,L)$ and weight $\epsilon$ is given by
  \begin{equation} \label{eq:PCnew} P_{C,new}(K,L,\epsilon) = \frac{\epsilon}{m^2} \cdot A_{K,L} \cdot N_{K,L}^\epsilon
    \cdot (p + \xi)^\epsilon.
  \end{equation}
\end{theorem}
\begin{myremark}
  The analysis of \theoref{theo:contribution} is exact for $\xi =0$ if no undetected error events occur in the decoding
  process. The parameter $\xi$ is introduced in \cite{stairfec} to heuristically adjust for such undetected error
  events.
\end{myremark}
While this offers exact results, the computational complexity is substantial, growing rapidly with $K$ and $L$. In the
following, a method of approximating the exact number via simulations is presented. The approximation in
\cite{stairfec} is based on the fact that it is simple to find the number of matrices, such that only the conditions on
the rows or the columns of a stall pattern with $\epsilon = \epsilon_{\min}$ hold. However, only a fraction $\gamma$ of
these actually represent a stall pattern, \emph{i.e.}, fulfill the conditions on the columns. This fraction can be estimated by
choosing $K$ vectors of weight $r_i \geq t+1$ for the rows such that the weight is $\epsilon$ (e.g. by choosing $K$ random
vectors of weight $t+1$ and $\epsilon - K(t+1)$ additional positions randomly out of the remaining zeros) and checking if
$s_j \geq t+1 \; \forall \; j \in [1,L]$ holds for every column. Formally, let $\mathcal{S}$ be a set of random binary
$K\times L$ matrices with row weight $r_i \geq t+1 \; \forall \; i \in [1,K]$ and $\sum_{i=1}^K r_i = \epsilon$. Then
\begin{equation*}
  \gamma \approx \hat{\gamma} = \frac{ |\{ S \in \mathcal{S} | s_j\geq t+1 \; \forall \; j \in [1,L]\}|}{|\mathcal{S}|}.
\end{equation*}
For $\epsilon = \epsilon_{\min}$ the exact number of stall patterns can then be estimated by
\begin{equation} \label{eq:Ntildemin} N_{K,L}^{\epsilon_{min}} \approx \hat{\gamma} \, \hat{N}_{K,L}^{\epsilon_{\min}} ,
\end{equation}
where $\hat{N}_{K,L}^{\epsilon} $ is given in \eqref{eq:amerrdist}.  Increasing the number of simulated matrices lets
this approximation approach the exact value, but only for stall patterns with $\epsilon = \epsilon_{\min}$. For the
estimation of \cite{stairfec} with larger $\epsilon$, not only the condition on the columns is neglected, but it is also
assumed, that randomly inserting the remaining the $\epsilon - \epsilon_{\min}$ errors always yields an unique
distribution. Employing \eqref{eq:Ntildemin} for $\epsilon > \epsilon_{\min}$, by performing the same steps and
inserting the remaining errors randomly, would not converge towards the exact value, but give an estimation of an upper
bound.  \lem{lem:cardZ} in the appendix allows for calculating the exact number of $K\times L$ matrices satisfying the
conditions on the rows by
\begin{equation} \label{eq:Ntilde} \tilde{N}_{K,L}^\epsilon = F(\epsilon,K).
\end{equation}
An approximation of the percentage $\gamma$ for which the conditions on the
columns hold, \emph{i.e.}, that represent a stall pattern, allows for the exact number to be approximated by
\begin{equation} \label{eq:gammaN} N_{K,L}^\epsilon \approx \hat{\gamma} \,\tilde{N}_{K,L}^\epsilon .
\end{equation}
Note that the complexity of this approach is much lower than the trivial approach of testing random $K \times L$ binary
matrices of weight $\epsilon$ on whether both the conditions on the rows and columns hold and dividing by the total
number $\mybinom{KL}{\epsilon}$ of such matrices.

\begin{figure}
%
%
\definecolor{mycolor1}{rgb}{0.00000,0.44700,0.74100}%
\definecolor{mycolor2}{rgb}{0.85000,0.32500,0.09800}%
\begin{tikzpicture}

\begin{axis}[%
width=0.85\textwidth,
height=0.4\textheight,
at={(0in,0in)},
scale only axis,
xmin=18,
xmax=36,
xlabel={$\text{Weight }\epsilon$},
ymode=log,
ymin=1,
ymax=10000000000000,
yminorticks=true,
ylabel={(Estimated) number of stall patterns},
axis background/.style={fill=white},
legend style={legend cell align=left,align=left,draw=white!15!black, at={(0.1,0.1)},anchor=south west}
]
\addplot[color=mycolor1,solid,mark=triangle]
  table[row sep=crcr]{%
18 64000000 \\ 
19 1152000000 \\ 
20 9792000000 \\ 
21 52224000000 \\ 
22 195840000000 \\ 
23 548352000000 \\ 
24 1188096000000 \\ 
25 2036736000000 \\ 
26 2800512000000 \\ 
27 3111680000000 \\ 
28 2800512000000 \\ 
29 2036736000000 \\ 
30 1188096000000 \\ 
31 548352000000 \\ 
32 195840000000 \\ 
33 52224000000 \\ 
34 9792000000 \\ 
35 1152000000 \\ 
36 64000000 \\ 
};
\addlegendentry{$\hat{N}_{6,6}^\epsilon$};

\addplot[color=red,solid,mark=square]
  table[row sep=crcr]{%
18 64000000 \\ 
19 288000000 \\ 
20 655200000 \\ 
21 991200000 \\ 
22 1110150000 \\ 
23 973125000 \\ 
24 689600625 \\ 
25 402952500 \\ 
26 196449750 \\ 
27 80390500 \\ 
28 27648495 \\ 
29 7966440 \\ 
30 1907556 \\ 
31 374256 \\ 
32 58815 \\ 
33 7140 \\ 
34 630 \\ 
35 36 \\ 
36 1 \\
};
\addlegendentry{$\tilde{N}_{6,6}^\epsilon$};

\addplot[color=green,solid,mark=o]
table[row sep=crcr]{%
	18 297200	\\
	19 6840000	\\
	20 44789040	\\
	21 143675040\\
	22 279581400\\
	23 369435600\\
	24 354910350\\
	25 259181640\\
	26 148284540\\
	27 67832200	\\
	28 25131150	\\
	29 7588800	\\
	30 1867320	\\
	31 371520	\\
	32 58725	\\
	33 7140		\\
	34 630		\\
	35 36		\\
	36 1		\\
};
\addlegendentry{$N_{6,6}^\epsilon$ };

\addplot[color=black,solid,mark=x]
table[row sep=crcr]{%
	18 384000.0000000\\
	19 6912000.0000000\\
	20 40622400.0000000\\
	21 154627200.0000000\\
	22 244233000.0000000\\
	23 347405625.0000000\\
	24 358592325.0000000\\
	25 247009882.5000000\\
	26 147140862.7500000\\
	27 65679038.5000000\\
	28 24938942.4900000\\
	29 7576084.4400000\\
	30 1857959.5440000\\
	31 371261.9520000\\
	32 58815.0000000\\
	33 7140.0000000\\
	34 630.0000000\\
	35 36.0000000\\
	36 1.0000000\\
};
\addlegendentry{$\hat{\gamma}\,\tilde{N}_{6,6}^\epsilon$ };

\addplot [color=black,dashed]
  table[row sep=crcr]{%
18 64000000\\
36 64000000\\
};
\addlegendentry{$\binom{L}{t+1}^K = \binom{6}{3}^6$};
\end{axis}
\end{tikzpicture}%
  \caption{Comparison between upper bounds and accurate number of stall patterns $N_{6,6}^{\epsilon}$ for $K=L=6$ and
    $t=2$, where $\hat{N}_{6,6}^{\epsilon}$ is given in \eqref{eq:amerrdist} and $\tilde{N}_{6,6}^{\epsilon}$ is given
    in \eqref{eq:Ntilde}. The low complexity method of \eqref{eq:gammaN} labeled
    $\hat{\gamma} \, \tilde{N}_{6,6}^{\epsilon}$ is shown, where $\hat{\gamma}$ is obtained from evaluating $1000$
    matrices for each point.}
  \label{fig:estcomp66}
\end{figure}
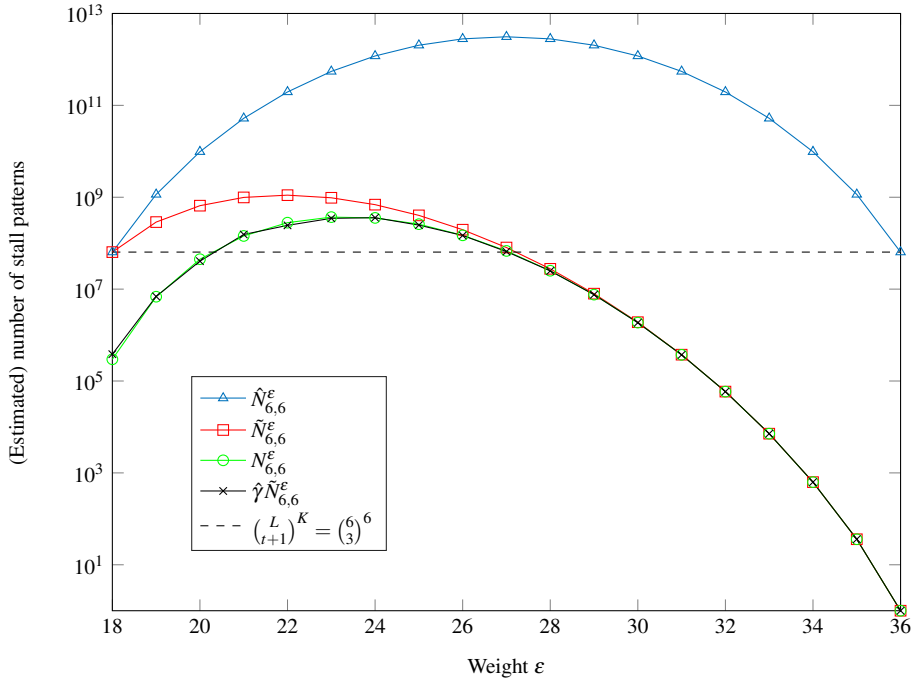

\fig{fig:estcomp66} gives a comparison for $K=6$, $L=6$ and $t=2$ between the upper bounds of \eqref{eq:amerrdist} and
\eqref{eq:amDepsK} as well as the exact value obtained by \eqref{eq:amDepsKex} and the approximation of the exact value
given in \eqref{eq:gammaN}. As expected, $\hat{N}_{K,L}^\epsilon$ given by \eqref{eq:amerrdist} and
$\tilde{N}_{K,L}^\epsilon$ given by \eqref{eq:Ntilde} result in the same value for $\epsilon_{\min}$, but for large
$\epsilon$ only $\tilde{N}_{K,L}^\epsilon$ approaches the correct value $N_{K,L}^\epsilon$, while
$\hat{N}_{K,L}^{\epsilon}$ overestimates $N_{K,L}^{\epsilon}$ significantly. Further it can be seen that the
approximation of the exact value obtained by \eqref{eq:gammaN} is very close to the exact numbers, even testing
only $1000$ matrices per point.

\subsection{Simulation Results} \label{subsec:simres}

To show the improvement in performance of our new technique, performance conjectures based on the combination of
simulation and analytical analysis for a specific staircase code are presented. As the error floor for the parameters
proposed in \cite{stairfec} is already below $10^{-20}$ even without resolving any stall patterns, the results are
given for different parameters, chosen such that the encoder and decoder can employ a simpler structure than
in~\cite{stairfec}, by using less memory and lower complexity component code decoders. The size of the blocks is quartered
compared to~\cite{stairfec} by using a~$[n=510,k=491]$ extended~$t=2$ error-correcting BCH code, resulting in block
size~$\frac{n}{2} \times \frac{n}{2} = 255 \times 255$. The corresponding rate for these parameters
is~$R = \frac{n-m}{k-m} = \frac{236}{255}$, which is slightly lower than the rate~$R_{[1]} = \frac{239}{255}$
of~\cite{stairfec}.

Decoding with the regular decoder, as described in~\sect{subsec:stcdecoding}, results in an error floor
at~$\sim 2 \cdot 10^{-10}$ for~$p = 5 \cdot 10^{-3}$, which is well above the desired error floor of~$10^{-15}$. This
error floor was found by simulation, using a sliding window of size~$W = 7$. The variable~$\xi$ adjusting for the
difference between estimated error floor and simulated probability of a stall pattern occurring was determined to
be~$\xi = 1.6 \cdot 10^{-3}$.

To find the capability of the improved decoder to solve stall patterns, a dedicated channel was implemented. Stall
patterns of a certain size and weight are inserted at a random position within two blocks. At the output, it is
determined whether the stall pattern was resolved or not. The basic assumption of this dedicated channel is that the
decoder is able to resolve all surrounding errors. For this reason and to avoid unwanted effects, such as resolving of a
stall pattern through an undetected error event, no errors other than the ones belonging to the stall pattern were
inserted. The simulation results are given in~\tab{tab:simres}.
\begin{table}
  \vspace{-10pt}
  \caption{Simulation and estimation results of the error floor for a staircase code with block size $255 \times 255$
    and $t=2$ error-correcting $[n=510,k=491]$ component code of distance $d_{\min}=6$. $K$ and $L$ give the number of
    rows and columns involved in the stall pattern and $\epsilon$ denotes the weight. The percentage of solvable stall
    patterns was obtained in simulation by applying the bit-flip operation as described in
    \protect\algo{algo:slidingresolve} on a dedicated channel designed to insert stall patterns. $P_{C,old}$ and
    $P_{C,new}$ give the approximation of the contribution to the error floor of \eqref{eq:PCold} and \eqref{eq:PCnew}
    respectively. $P_{floor,new}$ gives the contribution considering the percentage of stall patterns that can be solved
    for the given parameters and the error floor is the sum of this column.}
  \begin{center}
    \bgroup \def\arraystretch{1.15}
    \begin{tabular}{M{0.5cm}M{0.5cm}M{0.5cm}M{1.5cm}M{2cm}M{2cm}M{2cm}M{2cm}M{2cm}}
      \hline\noalign{\smallskip}
      K & L & \epsilon & \% \; \text{Solved} & P_{C,old}  & P_{C,new}  & P_{floor,new}\\ \hline\noalign{\smallskip}
      3 & 3 & 9 & 100 \% & 1.2 \cdot 10^{-10} &	1.2 \cdot 10^{-10} & 0\\
      3 & 4 & 12 & 51 \% & 4.1 \cdot 10^{-15} &	4.1 \cdot 10^{-15} & 2.0 \cdot 10^{-15}\\
      4 & 3 & 12 & 56 \% & 9.0 \cdot 10^{-15} &	9.0 \cdot 10^{-15} & 3.9 \cdot 10^{-15}\\
      4 & 4 & 12 & 100 \% & 1.4 \cdot 10^{-10} &	1.3 \cdot 10^{-11} & 0\\
      4 & 4 & 13 & 100 \% & 4.1 \cdot 10^{-12} &	3.9 \cdot 10^{-13} & 0\\
      4 & 4 & 14 & 79 \% & 4.4 \cdot 10^{-14} &	2.0 \cdot 10^{-15} & 4.4 \cdot 10^{-16}\\
      5 & 5 & 15 & 100 \% & 1.0 \cdot 10^{-10} &	2.2 \cdot 10^{-12} & 0\\
      5 & 5 & 16 & 99.9 \% & 7.5 \cdot 10^{-12} &	1.7 \cdot 10^{-13} & 6.8 \cdot 10^{-17}\\
      5 & 5 & 17 & 97.4 \% & 2.3 \cdot 10^{-13} &	3.6 \cdot 10^{-15} & 9.3 \cdot 10^{-17}\\
      5 & 5 & 18 & 95.1 \% & 4.4 \cdot 10^{-15} & 3.4 \cdot 10^{-17} & 1.7 \cdot 10^{-18}\\
      6 & 6 & 18 & 99.9 \% & 8.4 \cdot 10^{-11} & 3.9 \cdot 10^{-13} & 2.5 \cdot 10^{-16}\\
      6 & 6 & 19 & 99.9 \% & 1.0 \cdot 10^{-11} & 6.2 \cdot 10^{-14} & 7.9 \cdot 10^{-17}\\
      6 & 6 & 20 & 98.9 \% & 6.2 \cdot 10^{-13} &	2.8 \cdot 10^{-15} & 3.1 \cdot 10^{-17}\\
      7 & 7 & 21 & 100 \% & 7.3 \cdot 10^{-11} &	7.5 \cdot 10^{-14} & 0 \\
      7 & 7 & 22 & 99.9 \% & 1.4 \cdot 10^{-11} &	2.2 \cdot 10^{-14} & 1.4 \cdot 10^{-17}\\
      7 & 7 & 23 & 99 \% & 1.3 \cdot 10^{-12} &	1.7 \cdot 10^{-15} & 1.8 \cdot 10^{-17}\\
      \noalign{\smallskip}\hline
    \end{tabular}
    \egroup
  \end{center}
  \label{tab:simres}
  \vspace{-10pt}
\end{table}

In Columns 1-3 of Table~\ref{tab:simres}, the size and weight of the respective stall patterns are given. The fourth
column gives the percentage of stall patterns that the improved decoder was able to solve. For each size, more
than~$2000$ stall patterns were inserted as described above. If any errors in the corresponding blocks were observed at
the decoder output, the stall pattern was counted as unsolved. The remaining columns give the contribution to the error
floor, as found by applying the different estimations. The column labeled $P_{C,old}$ gives the estimated contribution
of stall patterns of the respective size, obtained by applying the estimation from~\cite{stairfec}, given
in~\eqref{eq:PCold}. In column $P_{C,new}$, the corresponding values obtained by~\eqref{eq:PCnew} are given. The column
$P_{floor,new}$ shows the contribution to the error floor for the improved decoder which resolves the respective
percentage of stall patterns.
The new dominating contributor are the stall patterns of size $(3,4)$ and $(4,3)$, mainly due to the inability of the decoder to solve a large percentage of these. An illustration of such an unsolvable stall pattern was given in \fig{fig:43stpa}. Our simulations show that the resolving strategy of only partially inverting large stall patterns (\emph{i.e.}, those which are not covered by Theorem~\ref{theo:guarantee-solving}) resolves a large percentage of these (e.g. $99.9 \%$ of $(6,6)$ stall patterns with $\epsilon_{\min} = 18$).

The error floor of the staircase code with the given parameters is expected to be at $\berout \approx 9 \cdot 10^{-15}$
for a BSC crossover probability of $p = 5 \cdot 10^{-3}$. For comparison, the error floor of the same staircase code
employing the regular decoder lies at $\berout \approx 2 \cdot 10^{-10}$. The performance of the code in terms of the
gap to BSC capacity for the given rate is estimated to $1$dB and it achieves a net coding gain (NCG) of $9.16$dB at a
$\berout$ of $10^{-15}$.
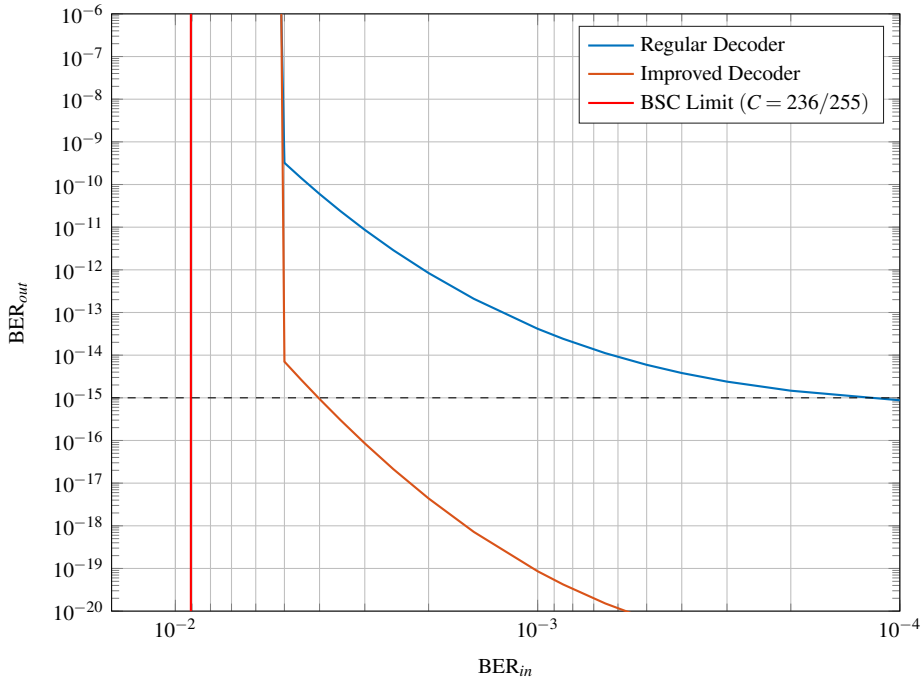
\begin{figure}
%
%
\definecolor{mycolor1}{rgb}{0.00000,0.44700,0.74100}%
\definecolor{mycolor2}{rgb}{0.85000,0.32500,0.09800}%
\begin{tikzpicture}

\begin{axis}[%
width=0.85\textwidth,
height=0.4\textheight,
at={(0.758in,0.481in)},
scale only axis,
x dir=reverse,
xmode=log,
xmin=0.0001,
xmax=0.015,
xminorticks=true,
xlabel={$\berin$},
xmajorgrids,
xminorgrids,
ymode=log,
ymin=1e-20,
ymax=1e-06,
ytick={1e-20, 1e-19, 1e-18, 1e-17, 1e-16, 1e-15, 1e-14, 1e-13, 1e-12, 1e-11, 1e-10, 1e-09, 1e-08, 1e-07, 1e-06},
ylabel={$\berout$},
ymajorgrids,
axis background/.style={fill=white},
legend style={legend cell align=left,align=left,draw=white!15!black}
]
\addplot [color=mycolor1,solid,thick]
  table[row sep=crcr]{%
0.0051	1e-06\\
0.005	3.22240228358985e-10\\
0.0045	1.42986397985259e-10\\
0.004	6.0111110634193e-11\\
0.0035	2.3708143083054e-11\\
0.003	8.65512691443788e-12\\
0.0025	2.86993773057473e-12\\
0.002	8.41183506738895e-13\\
0.0015	2.09330165665639e-13\\
0.001	4.15543695445261e-14\\
0.00085	2.41449435708062e-14\\
0.00065	1.11139841758619e-14\\
0.0005	5.93671978288915e-15\\
0.0004	3.81298155109986e-15\\
0.0003	2.39521798947046e-15\\
0.0002	1.46798918150775e-15\\
0.0001	8.7528921107261e-16\\
};
\addlegendentry{Regular Decoder};

\addplot [color=mycolor2,solid,thick]
  table[row sep=crcr]{%
0.0051	1e-06\\
0.005	7.04310242822496e-15\\
0.0045	2.67774340692394e-15\\
0.004	9.40153483337402e-16\\
0.0035	3.0027029907226e-16\\
0.003	8.55315867398588e-17\\
0.0025	2.11578892675201e-17\\
0.002	4.3796126506043e-18\\
0.0015	7.18826833831139e-19\\
0.001	8.61411416196282e-20\\
0.00085	4.20970076501034e-20\\
0.00065	1.5096090349327e-20\\
0.0005	6.57965171643522e-21\\
0.0004	3.65794247019023e-21\\
0.0003	1.97359725235181e-21\\
0.0002	1.03004636643944e-21\\
0.0001	5.18057902334645e-22\\
};
\addlegendentry{Improved Decoder};

\addplot [color=red,solid,thick]
  table[row sep=crcr]{%
0.00906217656129442	1e-20\\
0.00906217656129442	1e-06\\
};
\addlegendentry{BSC Limit $(C=236/255)$};

\addplot [color=black,dashed,forget plot]
  table[row sep=crcr]{%
0.015	1e-15\\
0.0001	1e-15\\
};
\end{axis}
\end{tikzpicture}%
  \caption{Conjectured performance comparison between the decoder given in \protect\cite{stairfec} and the improved
    decoder introduced in this work for a staircase code with block size $255 \times 255$ and $t=2$ error-correcting
    ${[n=510,k=491]}$ component code of distance $d_{\min}=6$. $\berin$ denotes the input bit error rate, \emph{i.e.}, the BSC
    crossover probability, and $\berout$ denotes the output bit error rate.}
  \label{fig:errflocomp}
\end{figure}
For comparison, the code presented in~\cite{stairfec} achieves a NCG of $9.41$dB and a gap to
capacity of $0.56$dB, however operating on much larger blocks.

\fig{fig:errflocomp} shows a comparison between the expected performance when using the regular decoder and the improved
decoder.  Note that the simulations on the ability of the improved decoder to resolve stall patterns were performed
under the assumption that the decoder is able to isolate all stall patterns, \emph{i.e.}, resolve all errors which are not part
of stall patterns. Furthermore, the assumption made in~\cite{stairfec} that stall patterns dominate the error floor, is
adopted. Assuming that the window size of the regular decoder is chosen such that only the first (lowest indexed) block
is free of all correctable errors, the sliding window size of the \emph{improved} decoder has to be increased by $3$
(see \sect{subsec:bitflipimplement}) to obtain a sufficient number of blocks that can be considered to be error-free
with the exception of stall patterns.

\section{Conclusion} \label{sec:conclusion}

Staircase codes are a powerful code construction for high-speed optical networks which perform close to the BSC
capacity for high rates. The decoding based on a hard decision component code provides efficient implementations, even at high data
rates. However, the usable block size is limited by requiring an error floor of $10^{-15}$ in optical communications. In
this work, an improved decoder was proposed, lowering the error floor significantly while increasing complexity only
marginally. This improvement enables the use of smaller block sizes at comparable rates which effectively lowers the
memory requirement and the component decoder complexity. Furthermore, an analysis of the error floor was presented
resulting in a more accurate estimation, which is especially important for the improved decoder.

Future work of interest includes an FPGA implementation capable of simulating the code with the given parameters down
to its estimated error floor to show that assumptions made on the capability of the decoder to correct errors
surrounding a stall pattern hold.

%

\begin{acknowledgements}
  The authors would like to thank Gerhard Kramer and Frank Kschischang for the valuable discussions and the anonymous reviewers for their comments that helped improve the quality and presentation of this work.
\end{acknowledgements}


\bibliographystyle{spmpsci} 
\bibliography{DCC.bib} 

\section*{Appendix}

\begin{lemma} \label{lem:cardZex}
  Denote by $r_1,...,r_K$ the weight of the rows and by $s_1,...,s_L$ the weight of the columns of a matrix $\mathbf{A}
\in \field{2}^{K\times L}$. The cardinality of the set
  \begin{equation}
    \label{eq:setA}
    \mathcal{Z}_{K,L}^{\epsilon} =    \left\{ \mathbf{A} | \wt(\mathbf{A}) = \epsilon, r_i \geq t+1 \; \forall \; i \in [1,K], s_i \geq t+1 \; \forall \; i \in [1,L]  \right\}
  \end{equation}
  is given by
  \begin{equation} \label{eq:amDepsKex} | \mathcal{Z}_{K,L}^\epsilon | = D(\epsilon,K,L,0)
  \end{equation}
  with
  \begin{equation}
    \label{eq:Dabcd}
    D(a,b,c,d) = \sum_{w_{b+d}=\max\left\{t+1,a-c(b-1)\right\}}^{\min\left\{a-(b-1)(t+1),c\right\}} D(a-w_{b+d},b-1,c,d)
  \end{equation}
  and
  \begin{align}
    D(w_{d+1},1,c,0) &= D(\epsilon,L,K,K), \label{eq:rec2}\\
    D(w_{d+1},1,c,K) &= \mathcal{A}([w_1,...,w_K],[w_{K+1},...,w_{K+L}]), \label{eq:recfin}
  \end{align}
  where $\mathcal{A}(\cdot,\cdot)$ is given in \cite{mat01}.
\end{lemma}

\begin{proof}
  The inputs of $D(a,b,c,d)$ give the number of non-zero positions $a$ that are to be distributed within the rows or
  columns (determined by an index offset $d$) of a matrix of size $b\times c$, such that the conditions of
  (\ref{eq:setA}) are fulfilled. Each recursion step of \eqref{eq:amDepsKex} determines the number of non-zero positions
  in one row, \emph{i.e.}, its Hamming weight.
  Trivially, the weight $r_i$ of a row of an $b \times c $ matrix is upper bounded by $r_i \leq c$. It follows that
  $r_b \geq a-c(b-1)$, because otherwise the resulting matrix cannot be of weight $a$, even if all remaining rows are of
  maximal weight. By definition $r_b \geq t+1$ and as this condition also has to hold for the other $b-1$ rows, the
  Hamming weight of the $b$-th row can be at most $a-(b-1)(t+1)$. Combining these conditions gives the limits of the sum
  in (\ref{eq:Dabcd}). The problem is reduced to distributing the remaining $a-r_b$ non-zero positions among the rows of
  a $b-1 \times c$ matrix, where the sum limits guarantee that the conditions on the weight of the rows can be met. When
  the case \eqref{eq:rec2} is reached, the remaining non-zero positions have to be in the last row and its weight $r_1$ is
  therefore given by the first function input. The recursion giving the weights $w_{K+1},..., w_{K+L}$ of the $L$
  columns is initiated, where $d=K$ gives the offset in the indices of the weights. The same arguments as above hold for
  the limits of the sum. When the last column is reached , \emph{i.e.}, the case \eqref{eq:recfin}, the function
  $\mathcal{A}([w_1,...,w_K],[w_{K+1},...,w_{K+L}])$ from \cite{mat01} (see also \cite{Perez2002}) gives the number of
  matrices corresponding to the unique weight vector $\mathbf{w}$ of size $K+L$.
\end{proof}
Note that while this notation offers a mathematically correct way to calculate the cardinality, when implemented it
would be beneficial to call the function $\mathcal{A}(\rv{},\sv{})$ as few times as possible, as it is the
computationally most complex part. As permutations of the input vectors lead to the same result, many calls of this
function can be avoided by, \emph{e.g.}, maintaining a look-up table.
\begin{lemma} \label{lem:cardZ} Denote by $r_1,...,r_K$ the weight of the rows of a matrix
  $\mathbf{A} \in \field{2}^{K\times L}$. The cardinality of the set
  \begin{equation*}
    \mathcal{Z}_{K,L}^{\epsilon} = \left\{ \matr{A}{K}{L} | \wt(\mathbf{A}) = \epsilon, r_i \geq t+1 \;\; \forall \; i \in [1,K]  \right\}
  \end{equation*}
  is given by
  \begin{equation} \label{eq:amDepsK} | \mathcal{Z}_{K,L}^\epsilon | = F(\epsilon,K)
  \end{equation}
  with
  \begin{align*}
    F(a,b) &= \sum_{j=\max\left\{t+1,a-L(b-1)\right\}}^{\min\left\{a-(b-1)(t+1),L\right\}} \mybinom{L}{j} F(a-j,b-1)
  \end{align*}
  and $F(a,1) = \mybinom{L}{a}$.

\end{lemma}

\begin{proof}
  There are $\mybinom{L}{h}$ unique binary vectors $\rv{}$ of length $L$ with $\wt(\rv{}) = h$. Summing over all
  $t+1 \leq h \leq L$ gives the number of unique binary vectors $\rv{}$ of length $L$, such that
  $t+1 \leq \wt(\rv{}) \leq L$, as
  \begin{equation}
    \sum_{j=t+1}^{L}\mybinom{L}{j} \; . \label{eq:onerow}
  \end{equation}
  Let there be $K$ vectors of weight
  \begin{align}
    t+1 \leq r_i &\leq L \; \forall \; i \in [1,K] \label{eq:whrK}
  \end{align} and let $\sum_{i=1}^K r_i = \epsilon$. For every distribution with $r_1 = \delta$, it holds that $\sum_{i=2}^K = \epsilon-\delta$, and by \eqref{eq:whrK} we obtain
  \begin{equation}
    \max \{t+1,\epsilon-L(K-1)\} \leq \delta \leq \min \{\epsilon -(K-1)(t+1),L\} \label{eq:recurbounds}
  \end{equation}
  as the region for $\delta$ for which the constraints on the other rows are still satisfiable.

  By setting the weight $r_1 = \delta$, the problem is reduced to finding the number of unique distributions of weight
  $\sum_{i=2}^K = \epsilon-\delta$ over $K-1$ rows, such that \eqref{eq:whrK} holds.  This results in the recursive
  equation
  \begin{align*}
    F(a,b) &= \sum_{j=\max\left\{t+1,a-L(b-1)\right\}}^{\min\left\{a-(b-1)(t+1),L\right\}} \mybinom{L}{j} D(a-j,b-1)
  \end{align*}
  and $F(a,1) = \mybinom{L}{a}$.  The partial term $\mybinom{L}{j}$ is given in \eqref{eq:onerow}, the bounds of the sum
  in \eqref{eq:recurbounds}. The partial term $F(a-j,b-1)$ results in the number of unique distributions of the reduced
  problem. When $b=1$ the last row is reached, which has to take the remaining difference between $\epsilon$ and the
  combined weight of the previous rows. The bounds on the sums assure, that the constraints for the last row are always
  met.
\end{proof}

\end{document}